\documentclass{article}
\usepackage[T1]{fontenc}
\usepackage{tgpagella}

\usepackage{a4}
\usepackage{array}
\usepackage{tabularx}
\usepackage{graphicx}
\usepackage{amsmath,amssymb,amsthm,mathtools}
\usepackage{paralist}
\usepackage{bm}
\usepackage{bbm}
\usepackage{xspace}
\usepackage{url}
\usepackage{fullpage, prettyref}
\usepackage{boxedminipage}
\usepackage{wrapfig}
\usepackage{ifthen}
\usepackage{color}
\usepackage{xcolor}
\usepackage{framed}
\usepackage[pagebackref,letterpaper=true,colorlinks=true,pdfpagemode=none,urlcolor=blue,linkcolor=blue,citecolor=violet,pdfstartview=FitH]{hyperref}
\usepackage[nameinlink]{cleveref}
\usepackage{fullpage}
\usepackage{esvect}
\usepackage{mdframed}
\usepackage{thmtools}
\usepackage{thm-restate}

\newtheorem{theorem}{Theorem}[section]

\newtheorem{lemma}[theorem]{Lemma}
\newtheorem{claim}[theorem]{Claim}
\newtheorem{corollary}[theorem]{Corollary}
\newtheorem{definition}[theorem]{Definition}

\newtheorem{observation}[theorem]{Observation}

\newtheorem{question}[theorem]{Question}

\newcommand{\ignore}[1]{}

\newcommand{\cA}{{\cal A}}

\newcommand{\cC}{{\cal C}}
\newcommand{\cD}{\mathcal{D}}
\newcommand{\cE}{{\cal E}}

\newcommand{\cR}{{\cal R}}

\newcommand{\poly}{\mathrm{poly}}

\newcommand{\ceil}[1]{\lceil#1\rceil}
\newcommand{\Exp}{\EX}
\newcommand{\floor}[1]{\lfloor#1\rfloor}

\newcommand{\EX}{\mathbf{E}}

\newcommand{\Sec}[1]{\hyperref[sec:#1]{\Cref*{sec:#1}}} 
\newcommand{\Eqn}[1]{\hyperref[eq:#1]{(\ref*{eq:#1})}} 
\newcommand{\Fig}[1]{\hyperref[fig:#1]{Fig.\,\ref*{fig:#1}}} 
\newcommand{\Tab}[1]{\hyperref[tab:#1]{Tab.\,\ref*{tab:#1}}} 
\newcommand{\Thm}[1]{\hyperref[thm:#1]{Theorem\,\ref*{thm:#1}}} 
\newcommand{\Fact}[1]{\hyperref[fact:#1]{Fact\,\ref*{fact:#1}}} 
\newcommand{\Lem}[1]{\hyperref[lem:#1]{Lemma\,\ref*{lem:#1}}} 
\newcommand{\Prop}[1]{\hyperref[prop:#1]{Prop.~\ref*{prop:#1}}} 
\newcommand{\Cor}[1]{\hyperref[cor:#1]{Corollary~\ref*{cor:#1}}} 
\newcommand{\Conj}[1]{\hyperref[conj:#1]{Conjecture~\ref*{conj:#1}}} 
\newcommand{\Def}[1]{\hyperref[def:#1]{Definition~\ref*{def:#1}}} 
\newcommand{\Alg}[1]{\hyperref[alg:#1]{Alg.~\ref*{alg:#1}}} 
\newcommand{\Obs}[1]{\hyperref[obs:#1]{Obs.~\ref*{obs:#1}}} 
\newcommand{\Ex}[1]{\hyperref[ex:#1]{Ex.~\ref*{ex:#1}}} 
\newcommand{\Clm}[1]{\hyperref[clm:#1]{Claim~\ref*{clm:#1}}} 
\newcommand{\Step}[1]{\hyperref[step:#1]{Step~\ref*{step:#1}}} 

\usepackage[utf8]{inputenc} 
\usepackage[T1]{fontenc}    
\usepackage{hyperref}       
\usepackage{url}            
\usepackage{booktabs}       
\usepackage{amsfonts}       
\usepackage{nicefrac}       
\usepackage{microtype}      
\usepackage{xcolor}         
\usepackage[ruled,linesnumbered]{algorithm2e}
\SetKwComment{Comment}{/* }{ */}
\usepackage{wrapfig}
\usepackage{algorithmic}
\usepackage{float}
\newcommand{\supp}{\mathsf{supp}}
\newcommand{\query}{\mathsf{count}}

\hypersetup{linktocpage = true, linkcolor = teal}

\title{Clustering with Non-adaptive Subset Queries\thanks{A preliminary version of this work appeared at NeurIPS 2024. This collaboration is the result of an EnCORE Institute workshop.}}

\author{%
  Hadley Black\thanks{CUNY Baruch College. This work was completed while the author was a Postdoctoral Researcher at UC San Diego supported by NSF TRIPODS Institute grant 2217058 (EnCORE) and NSF 2133484.}
\and
  Euiwoong Lee\thanks{University of Michigan. Supported in part by NSF grant 2236669 and Google. }
\and
  Arya Mazumdar\thanks{UC San Diego. Supported by NSF TRIPODS Institute grant 2217058 (EnCORE) and NSF 2133484.}
\and
  Barna Saha \thanks{UC San Diego. Supported by NSF TRIPODS Institute grant 2217058 (EnCORE) and NSF 2133484.} 
}

\begin{document}

\date{}
\maketitle

\begin{abstract}
Recovering the underlying $k$-clustering of a set $U$ of $n$ points by asking pair-wise \emph{same-cluster} queries has garnered significant interest in the past few years. Given a query $S \subset U$, $|S|=2$, the oracle returns {\em yes} if the points are in the same cluster and {\em no} otherwise. For \emph{adaptive} algorithms, the query complexity is known to be $\Theta(nk)$, while \emph{non-adaptive} algorithms are extremely limited: even for $k=3$, such algorithms require $\Omega(n^2)$ queries, matching the trivial upper bound. However, non-adaptivity is highly desirable since it allows queries to be asked in parallel. To break the quadratic barrier for non-adaptive queries, we study a natural generalization of this problem to {\em subset queries} for $|S|>2$, where the oracle returns the number of clusters intersecting $S$. Previous work obtained an $O(n)$ query adaptive algorithm, but the realm of \emph{non-adaptive} algorithms remained completely unknown.

In this paper, we give the first \emph{non-adaptive} algorithms for clustering with subset queries. Our main result is a non-adaptive algorithm making $O(n \log k \cdot (\log k + \log\log n)^2)$ queries, improving to $O(n \log \log n)$ when $k$ is constant. In addition to non-adaptivity, we make other practical considerations, such as enforcing a bound, $s$, on the query size. We show $\Omega(\max(n^2/s^2,n))$ queries are necessary and obtain algorithms making $\smash{\widetilde{O}(n^2k/s^2)}$ queries for any $s \leq \sqrt{n}$ and $\smash{\widetilde{O}(n^2/s)}$ queries for any $s \leq n$. Finally, we obtain improved upper bounds when the clusters are roughly balanced, and when the algorithm is allowed \emph{two rounds} of adaptivity.

\end{abstract}

\newpage

\tableofcontents
\newpage

\section{Introduction}

Clustering is one of the most fundamental problems in unsupervised machine learning, and permeates beyond the boundaries of statistics and computer science to social sciences, economics, and so on. The goal of clustering is to partition items so that similar items are in the same group. The applications of clustering are manifold. However, finding the underlying clusters is sometimes hard for an automated process due to data being noisy, incomplete, but easily discernible by humans. Motivated by this scenario, in order to improve the quality of clustering,  early works have studied the so-called  clustering under ``limited supervision'' (e.g., \cite{cohn2003semi,bair2013semi}). Balcan and Blum initiated the study of clustering under active feedback \cite{balcan2008clustering} where given the current clustering solution, the users can provide feedback whether a cluster needs to be merged or split. Perhaps a simpler query model would be where users only need to answer the number of clusters, and that too only on a subset of points without requiring to analyze the entire clustering. This scenario is common in unsupervised learning problems, where
a centralized algorithm aims to compute a clustering by crowdsourcing. The crowd-workers play the role of an oracle here, and are able to answer simple queries that involve a small subset of the universe (e.g.~\cite{wang2012crowder,vesdapunt2014crowdsourcing}). Learning a partition in a graph with edge-queries is a closely related problem, and had been studied in \cite{reyzin2007learning}.
Mazumdar and Saha \cite{mazumdar2017theoretical, MS17a,MS17b}, and in independent works Mitzenmacher and Tsourakis \cite{mitzenmacher2016predicting}, as well as Asthani, Kushagra and Ben-David \cite{ashtiani2016clustering} initiated various theoretical studies of clustering with pair-wise aka {\em same-cluster} queries (which are same as edge-queries). Given any pair of points $u,v$, the oracle returns whether $u$ and $v$ belong to the same cluster or not. This type of query is easy to answer and lends itself to simple implementations~\cite{mazumdar2017semisupervised}. This has been subsequently extremely well-studied in the literature, e.g.~\cite{saha2019correlation,del2022clustering,mazumdar2017theoretical,bressan2020exact,huleihel2019same}.
Moreover, clustering with pair-wise queries is intimately related to several well-studied problems such as correlation clustering \cite{BBC04, ACN08, CMSY15, saha2019correlation, CCLLNV24}, edge-sign prediction problem \cite{leskovec2010predicting, mitzenmacher2016predicting}, stochastic block model \cite{abbe2018community, mukherjee2024recovering}, etc.

Depending on whether there is an interaction between the
learner/algorithm and the oracle, the querying algorithms can be classified as adaptive and non-adaptive \cite{MS17a}. In adaptive querying, the learner can decide the next query based on the answers to the previous queries. An algorithm is called \emph{non-adaptive} if all of its queries can be specified in one round. Non-adaptive algorithms can parallelize the querying process as they decide the entire set of queries apriori. This may greatly speed up the algorithm in practice, significantly reducing the time to acquire answers \cite{gu2012towards}. Thus, in  a crowdsourcing setting being non-adaptive is a highly desirable property. On the flip side, this makes non-adaptive algorithms significantly harder to design.
In fact, when adaptivity is allowed, $\Theta(nk)$ pair-wise queries are both necessary and sufficient to recover the entire clustering, where $n$ is the number of points in the ground set to be clustered and $k$ (unknown) is the number of clusters. However as shown in \cite{MS17a} and our \Cref{thm:3-2-NA-LB}, even for $k = 3$, even randomized \emph{non-adaptive algorithms can do no better than the trivial $O(n^2)$ upper bound attained by querying all pairs}. 



We study a generalization of pair-wise queries to subset queries, where given any subset of points, the oracle returns the number of clusters in it. We consider the problem of recovering an unknown $k$-clustering (a partition) on a universe $U$ of $n$ points via black-box access to a \emph{subset query oracle}. More precisely, we assume that there exists a groundtruth partitioning of $U = \bigsqcup_{i=1}^k C_i$, and upon querying with a subset $S \subseteq U$, the oracle returns $\query(S) = |\{i: C_i \cap S \neq \emptyset\}|$, the number of clusters intersecting $S$. Considering the limitations of pair-wise queries for non-adaptive schemes, we ask the question if it is possible to use subset queries to design significantly better non-adaptive algorithms. We remark that subset queries as a generalization of pairwise queries are also actively studied in the context of estimation problems in graphs (e.g. \cite{BHRRS20, LW20, DBLP:journals/corr/abs-2603-14306}).

In addition to being a natural model for interactive clustering, this problem also falls into the growing body of work known as \emph{combinatorial search} \cite{A88,DH00} where the goal is to reconstruct a hidden object by viewing it through the lens of some indirect query model (such as group testing~\cite{du2000combinatorial,HS87, DH00,flodin2021probabilistic,mazumdar2016nonadaptive,PR08}). The problem is also intimately connected to coin weighing where given a hidden vector $x \in \{ 0, 1 \}^n$, the goal is to reconstruct $x$ using queries of the form $q(S) := \sum_{i \in S} x_i$ for $S \subseteq [n]$. It is known that $\Theta(n / \log n)$ is the optimal number of queries
        ~\cite{lindstrom1964combinatory, lindstrom1966combinatory, cantor1966determination}, which can be obtained by a non-adaptive algorithm. There are improvements for the case when $\| x \|_1 = d$ for $d \ll n$~\cite{bshouty2009optimal, BM11a, BM11b}. Moreover, there has been significant work on graph reconstruction where the task is to reconstruct a hidden graph $G = (V, E)$ from queries of the form $q(S, T) := |\{ (u, v) \in E: u \in S, v \in T \}|$ for subsets $S, T \subseteq V$ ~\cite{AC04, choi2008optimal, mazzawi2010optimally, choi2013polynomial}.
        There are also algorithms that perform certain tasks more efficiently than learning the whole graph (sometimes using different types of queries)~\cite{rubinstein2017computing,MN20, assadi2021graph, AEGLMN22, liao2024learning, BHRRS20, LW20, AMM22}, and quantum algorithms that use fewer queries than classical algorithms~\cite{lee2021quantum}.

It is not too difficult to show that an algorithm making $O(n \log k)$ queries (\Cref{sec:simple-adaptive}) is possible for $k$-clustering, while $\Omega(n)$ queries is an obvious information theoretic lower bound since each query returns $\log k$ bits of information and the number of possible $k$-clusterings is $k^n = 2^{n \log k}$. In fact, there exists an algorithm with $O(n)$ query complexity \cite{CL24}.
 However, both of these algorithms are adaptive, ruling them out for the non-adaptive setting. So far, the non-adaptive setting of this problem has remained unexplored.

\subsection{Results}

Our main results showcase the significant strength of using subset queries in the non-adaptive setting. We give randomized algorithms that recover the exact clustering with probability $1-\delta$, for any constant $\delta > 0$ using only near-linear number of subset queries. For the purpose of clarity in the introduction, our theorems are stated for small constant $\delta$, and the exact dependence on this parameter can be seen in the formal theorem statements in the technical sections. Though our focus in this work is on \emph{query complexity}, we remark that all of our algorithms are computationally efficient and run in polynomial time\footnote{For the sake of brevity and to not distract from the main message, we do not include runtime analysis in our proofs or in our theorem statements. We hope that it is clear from our pseudocode that each algorithm can be implemented in polynomial time.}. A tabular view of our results is given in table \ref{table:results}.

\paragraph{Algorithms with unrestricted query size.} Our main result is the following algorithm for $k$-clustering with subset queries of unrestricted size.

\begin{theorem} [\Cref{thm:1}, informal] \label{thm:1-informal} There is a randomized, non-adaptive $k$-clustering algorithm making $O(n \log k \cdot (\log k + \log \log n)^2)$ subset queries. \end{theorem}

Additionally, we obtain the following algorithm which improves the dependence on $n$, at the cost of a worse dependence on $k$. In particular, for constant $k$ this algorithm has query complexity $O(n \log \log n)$, improving on \Cref{thm:1-informal} by a $O(\log \log n)$ factor.

\begin{theorem} [\Cref{thm:nloglogn}, informal] \label{thm:nloglogn-informal} There is a randomized, non-adaptive $k$-clustering algorithm making $O(n \log \log n \cdot k \log k)$ subset queries.
\end{theorem}

\noindent{\bf Bounding query size:} Another practical consideration is the allowed query size, $s$.  Depending on the scenario, and capabilities of the oracle, it may be easier to handle queries on small subsets. An extreme case is \emph{pair-wise queries} ($s=2$), 
where $O(nk)$ pair queries are enough with adaptivity, but any non-adaptive algorithm has to use $\Omega(n^2)$ queries even for $k = 3$. 
Since a subset query on $S$ can be simulated by $\smash{{|S| \choose 2}}$ pair queries, 
we immediately get an $\Omega(n^2/s^2)$ lower bound. Combining with the $\Omega(n)$ information-theoretic lower bound gives the following theorem.

\begin{theorem} [\Cref{cor:3-s-LB}, restated] \label{thm:LB-NA} Any non-adaptive $k$-clustering algorithm that is only allowed to query subsets of size at most $s$ must make at least $\smash{\Omega(\max(n^2/s^2, n))}$ queries. \end{theorem}

\Cref{thm:1-informal,thm:nloglogn-informal} above show that this can be bypassed by allowing larger subset queries. However, some of these queries are of size $\Omega(n)$, and this raises the question, \emph{is there a near-linear non-adaptive algorithm which only queries subsets of size at most $O(\sqrt{n})$?} We answer this in the affirmative, implying that our lower bound is tight in terms of $s$.\footnote{Since the preliminary version of this work, \cite{DBLP:conf/colt/BlackMS25} obtained a non-adaptive $\widetilde{O}(n^2/s^2)$ query algorithm, subsuming \Cref{thm:bounded-2-informal}.}

\begin{theorem} [\Cref{thm:bounded-2}, informal] \label{thm:bounded-2-informal} For any $s \in [2,\sqrt{n}]$, there is a non-adaptive $k$-clustering algorithm making $\widetilde{O}(n^2k/s^2)$ subset queries of size at most $s$. \end{theorem}

In particular, for constant $k$, the above algorithm matches the lower bound of \Cref{thm:LB-NA} up to $\log n$ factors. We leave the question of improving the dependence on $k$ as an interesting open problem. Using completely different techniques we obtain the following alternative algorithm with improved dependence on $k$, but worse dependence on $s$. 

\begin{theorem} [\Cref{thm:bounded-1}, informal] \label{thm:bounded-1-informal} For any $s \in [2,n]$, there is a randomized, non-adaptive $k$-clustering algorithm making $\widetilde{O}(n^2/s)$ subset queries of size at most $s$. \end{theorem}






\paragraph{The balanced case.} Next, we consider the natural special case of recovering a $k$-clustering when the cluster sizes are close to being balanced, within some multiplicative factor $B \geq 1$. We call such a clustering $B$-balanced. Balanced clustering is a well-studied assumption, notably in the area of image segmentation~\cite{shi2000normalized}.

For brevity in the introduction we state our theorems for this setting with constant $B$, and the exact dependence on this parameter can be seen in \Cref{sec:balanced}.  

\begin{theorem} [\Cref{thm:k-bal-1,thm:k-bal-2}, informal] \label{thm:bal-informal} There are non-adaptive algorithms for recovering an $O(1)$-balanced $k$-clustering making (a) $O(n \log k)$ subset queries if $\smash{k \leq O(\frac{n}{\log^3 n})}$, and (b) $O(n \log^2 k)$ subset queries for any $k \leq n$. \end{theorem}

\paragraph{Allowing two rounds of adaptivity.} Lastly, we show that if we allow an extra round of adaptivity, then we obtain the following algorithms with improvements to the logarithmic factors. 

\begin{theorem} [\Cref{thm:2-round,thm:2-round-bal}, informal] \label{thm:2round-informal} There is a  $2$-round deterministic $k$-clustering algorithm making $O(n \log k)$ subset queries. There is a randomized $2$-round algorithm for recovering an $O(1)$-balanced $k$-clustering making $O(n \log \log k)$ subset queries. \end{theorem}


\begin{table*}[t]
\label{table:results}
\centering
\begin{tabular}{|c|c|c|c|}
\hline
 & \textbf{1-round} & \textbf{2-round} & \textbf{1-round}, $s \leq \sqrt{n}$  \\
\hline
\textbf{General} &
$O(n \log k \cdot (\log k + \log\log n)^2)$ (\ref{thm:1-informal}) & $O(n \log k)$ (\ref{thm:2round-informal}) &
$\widetilde{O}(n^2/s)$ (\ref{thm:bounded-2-informal}) \\[5pt]
&
$O(n \log \log n \cdot k\log k)$ (\ref{thm:nloglogn-informal}) & &
$\widetilde{O}(n^2k/s^2)$ (\ref{thm:bounded-1-informal})
 \\[5pt]
&
$\Omega(n)$ \emph{(info-theoretic)} & $\Omega(n)$ \emph{(info-theoretic)} &
$\Omega(n^2/s^2)$ (\ref{thm:LB-NA}) \\[5pt]
\hline
\textbf{Balanced} &
$O(n \log k + k\log^4 k)$ (\ref{thm:bal-informal}) &
$O(n\log \log k)$ (\ref{thm:2round-informal}) &
 \\[5pt]
&
$\Omega(n)$ \emph{(info-theoretic)} &
$\Omega(n)$ \emph{(info-theoretic)} &
 \\[5pt]
\hline
\end{tabular}
\caption{\small{Table of our results. The last column lists our results regarding algorithms that are only able to query subsets of size at most $s$. We remark that after the preliminary version of this work, the non-adaptive query complexity for the bounded query setting was resolved (up to logarithmic factors) by \cite{DBLP:conf/colt/BlackMS25}, who found an algorithm using $\widetilde{O}(n^2/s^2)$ queries via different ideas. More generally, \cite{DBLP:conf/colt/BlackMS25} also gave smooth trade-offs with round complexity in the bounded query setting. We also remind the reader that there is an adaptive $O(n)$ query algorithm (using size $O(k)$ queries) due to \cite{CL24}, matching the information-theoretic lower bound. Our results for learning "balanced" partitions are stated here for the case when cluster sizes are the same up to a constant factor. See the formal theorem statements for the dependence on the balance parameter. }}
\end{table*}

\paragraph{Organization.} The remainder of the paper is organized as follows. In \Cref{sec:tech} we describe our main ideas and techniques and give high level descriptions of how they lead to our main results. We discuss some open questions in \Cref{sec:discussion}. \Cref{sec:unbounded} contains formal proofs of our main results for unrestricted query size, \Cref{thm:1-informal,thm:nloglogn-informal}. In \Cref{sec:bounded} we give formal proofs of our results for bounded query size, \Cref{thm:bounded-1-informal,thm:bounded-2-informal} as well as our lower bound \Cref{thm:LB-NA}. Finally, we prove our results for the balanced setting in \Cref{sec:balanced} and our results for two-round algorithms in \Cref{sec:2round}.

\subsection{Ideas and Techniques} \label{sec:tech}

In this section we describe our main ideas and give high-level descriptions of how they lead to our main results. We start with a discussion of the approach taken by \cite{CL24} to obtain an optimal adaptive algorithm.

\paragraph{High-level description of \cite{CL24} algorithm.} Before diving into our techniques, we briefly describe how \cite{CL24} obtain their $O(n)$ query adaptive algorithm and the apparent roadblock in making their approach non-adaptive. Using their terminology, call $I \subseteq U$ an \emph{independent set} if every element of $I$ belongs to a distinct cluster. The high-level strategy is to start with $n$ singleton independent sets, which are iteratively "merged". When $I_1,I_2$ are merged, they are replaced by a new independent set, $I_3$, containing a single element from each cluster represented in $I_1 \cup I_2$ and the remaining elements are set aside (the algorithm has found a representative for their clusters). This merge can be viewed as the reconstruction of a bipartite graph with an edge from $u \in I_1$ to $v \in I_2$ if $u,v$ are in the same cluster. In fact, since $I_1,I_2$ are independent sets, this graph is a \emph{matching}. This allows the subset query (called "rank" queries in \cite{CL24}) to be used to simulate \emph{additive} (edge-count) queries in this graph: for $S \subseteq I_1$, observe that 
\[
\big(\# \text{ clusters intersected by } S \cup I_2\big) = |S| + |I_2| - \big(\# \text{ edges from } S \text{ to } I_2\big)\text{.}
\]
Using this observation, one can leverage known graph reconstruction algorithms using additive queries. 

Note that this simulation crucially relies on the fact $I_1,I_2$ are independent sets and so the graph being reconstructed is a \emph{matching}. In particular, this critically relies on adaptivity to discard repeated elements from the same cluster so that this query simulation is possible. In the non-adaptive setting that we consider, it is not clear how to salvage this idea. As a result, our approach differs from \cite{CL24} in that we do not attempt to simulate graph reconstruction algorithms directly, as doing so seems to require adaptivity. Nonetheless, we are still able to make use of group testing primitives, which also are foundational in graph reconstruction. 

\paragraph{Roadmap.} Our techniques broadly fall under two categories which we describe in the following subsections. In \Cref{sec:tech-1} we describe a connection that our problem has with non-adaptive \emph{group testing}. This technique alone (essentially) leads to \Cref{thm:bounded-1-informal}, and in particular an $O(n \cdot\mathrm{poly}(\log n))$ query non-adaptive algorithm. Improving to $O(n \cdot\mathrm{poly}(\log\log n))$ requires a second idea enabling the use of large random subsets to discover small clusters, which we describe in \Cref{sec:tech-2}. In \Cref{sec:tech-3}, we describe how to combine these ideas to prove our main result,  \Cref{thm:1-informal}. 


\paragraph*{Warm up.} 
When there are only $2$ clusters, there is a trivial non-adaptive algorithm making $O(n)$ \emph{pair queries}: Choose an arbitrary $x \in U$ and query $\{x,y\}$ for every $y \in U$. The set of points $y$ where $\query(\{x,y\}) = 1$ form one cluster, and the second cluster is the complement. If we allow one more round of adaptivity, then for $3$-clustering we could repeat this one more time and again get an $O(n)$ query algorithm. However, for \emph{non-adaptive} $3$-clustering it is impossible to do better than the trivial $O(n^2)$ algorithm (see \Cref{thm:3-2-NA-LB}). Essentially, this is because in order to distinguish the clusterings $(\{x\},\{y\},U \setminus \{x,y\})$ and $(\{x,y\}, \emptyset, U \setminus \{x,y\})$ the algorithm must query $\{x,y\}$ and there are ${n \choose 2}$ ways to hide this pair. Overcoming this barrier using subset queries requires significant new ideas. 

\subsubsection{A Connection with Combinatorial Group Testing} \label{sec:tech-1}

Given a universe $U = \{x_1,\ldots,x_n\}$ of $n$ points, recall that our goal is to reconstruct a hidden $k$-partition $C_1 \sqcup \cdots \sqcup C_k = U$ where we are allowed to query a set $S \subseteq U$ and receive back $\query(S) := |\{i \in [k] \colon S \cap C_i \neq \emptyset\}|$. For each cluster $i \in [k]$, consider the corresponding Boolean vector $\smash{v^{(i)} \in \{0,1\}^n}$ defined as $\smash{v^{(i)}_j = \boldsymbol{1}(x_j \in C_i)}$. Let us use $\smash{\supp(v^{(i)}) = \{j \colon v^{(i)}_j = 1\}}$ to denote the support of such a vector and note that $\smash{(\supp(v^{(i)}) \colon i \in [k])}$ forms a partition $[n]$. An equivalent formulation of our problem is thus to reconstruct $\smash{v^{(1)},\ldots,v^{(k)}}$. The classic problem of \emph{combinatorial group testing} considers the problem of reconstructing a Boolean vector $v \in \{0,1\}^n$ via $\mathsf{OR}$ queries: given $S \subseteq [n]$, the query $\mathsf{OR}_S(v)$ returns whether or not $v$ has a $1$-valued coordinate in $S$. Classic results in group testing show that efficient non-adaptive algorithms for this problem exist. The following observation connects subset queries for $k$-clustering and $\mathsf{OR}$ queries for group testing. We use $\boldsymbol{1}(\cdot)$ to denote the Boolean indicator for the event $(\cdot)$. 

\begin{restatable}{observation}{subsetOR} \label{obs:subset->OR}Consider an arbitrary point $x \in U$ and let $C_i$ be the cluster containing it. Then 
\[
\mathsf{OR}_S(v^{(i)}) = \boldsymbol{1}(\query(S) = \query(S \cup \{x\})) \text{.}
\]
\end{restatable}

The reason for this observation is simply that since $x \in C_i$, the set $S$ intersects $C_i$ if and only if the subset query does not increase when $x$ is added to $S$. This simple observation allows us to leverage group testing procedures as black-box subroutines in our algorithms. In the following subsections, we describe how we obtain algorithms for our problem using two group testing lemmas as basic primitives. 

\paragraph{Obtaining a simple $\widetilde{O}(n)$ algorithm.} 

Our first group testing primitive comes from the following lemma which we prove in \Cref{sec:supp-recovery}.

\begin{restatable}{lemma}{tER} \label{lem:t-ER} Let $v \in \{0,1\}^n$. There is a non-adaptive algorithm that makes $O(t \log (n/\alpha))$ $\mathsf{OR}$ queries on subsets of size $\ceil{n/t}$, and if $|\supp(v)| \leq t$, returns $\supp(v)$ with probability $1-\alpha$, and otherwise certifies that $|\supp(v)| > t$ with probability $1$. 
\end{restatable}

In particular, combining this lemma with \Cref{obs:subset->OR} immediately yields the following corollary.

\begin{corollary} \label{cor:t-ER-cor} Given a $k$-clustering on $U$ of size $n$ and an element $x \in U$, let $C(x)$ denote the cluster containing $x$. There is a non-adaptive algorithm that makes $O(t \log (n/\alpha))$ subset queries of size at most $\ceil{n/t} + 1$, and if $|C(x)| \leq t$, returns $C(x)$ with probability $1-\alpha$, and otherwise certifies that $|C(x)| > t$ with probability $1$.\end{corollary}

\begin{proof} Let $U = \{x_1,\ldots,x_n\}$ denote a fixed, arbitrary ordering on $U$ and let $v \in \{0,1\}^n$ denote the Boolean vector such that $v_i = 1$ iff $x_i \in C(x)$. By \Cref{obs:subset->OR} we can simulate the algorithm of \Cref{lem:t-ER} using only twice as many subset queries, each of which has size at most $\ceil{n/t} + 1$. Moreover, clearly if the algorithm outputs $\supp(v)$, then we can return $C(x)$, and if the algorithm certifies that $|\supp(v)| > t$, then we can certify that $|C(x)| > t$. \end{proof}

The above corollary tells us the following. Consider a cluster $C$ of size $|C| = \gamma n$ for $\gamma \in (0,1]$, so that a uniform random point $x$ from $U$ will land in $C$ with probability $\gamma$. With this $x$ in hand, we can leverage \Cref{cor:t-ER-cor} to learn $C$ with probability $1-\alpha$ using $O(\gamma n \log (n/\alpha))$ non-adaptive queries, if we know a constant factor upper bound on $|C|$. This suggests the following approach. For each $p \in [\log n]$, try to learn the clusters of size $\approx \frac{n}{2^p}$ by sampling $\approx 2^p$ independent uniform $x$'s and running the procedure from \Cref{cor:t-ER-cor} with $t$ set to $\approx \frac{n}{2^p}$. Setting $\alpha = 1/\poly(n)$, this results in an algorithm making $\approx n \log^2 n$ queries. Boosting the number of $x$'s that are chosen in each stage by a factor of $O(\log k)$ allows us to show that all $k$ clusters are learned with high probability, resulting in an $O(n \log^2 n \log k)$ query non-adaptive algorithm.

In \Cref{sec:bounded-1}, we give a formal proof of this result. In fact, we obtain a more general algorithm (see \Cref{thm:bounded-1}), which queries subsets of size at most $s$ and has query complexity $O(n \log n \log k(\frac{n}{s} + \log s))$. This generalization comes by observing that for $s < |S|$, a single $\mathsf{OR}$ query on $S$ can be simulated by $\ceil{|S|/s}$ $\mathsf{OR}$ queries on sets of size at most $s$ simply by breaking up $S$ into smaller sets. This yields a simple extension of the group testing primitive \Cref{lem:t-ER} and consequently of \Cref{cor:t-ER-cor}. 


\subsubsection{Discovering Small Clusters with Large Queries} \label{sec:tech-2}

Our second idea uses only elementary probability theory and exploits the connectivity threshold of random graphs. The idea is best explained by focusing on the case of $3$-clustering. Recall that using \emph{pair queries only}, there is a trivial $O(n)$ non-adaptive algorithm for $2$-clustering, but $\Omega(n^2)$ non-adaptive pair queries are needed for $3$-clustering (\Cref{thm:3-2-NA-LB}). Thus, the $k=3$ setting of the problem is already interesting for non-adaptive algorithms.

It suffices to correctly reconstruct the two largest clusters, since the third cluster is just the complement of their union. Let $A,B$ denote the largest, and second largest clusters, respectively. Since $|A| \geq n/3$, it is easy to find: sample a random $x \in U$ and query $\{x,y\}$ for every $y \in U$. The cluster containing $x$ is precisely $\{y \in U \colon \query(\{x,y\}) = 1\}$. With probability at least $1/3$, we have $x \in A$ and so repeating this a constant number of times will always recover $A$. On the other hand, $B$ may be arbitrarily small and in this case the procedure clearly fails to recover it. 
The first observation is that once we know $A$, we can exploit larger subset queries to explore $U \setminus A$ since $\query(S \setminus A) = \query(S) - \mathbf{1}(S \cap A \neq \emptyset)$. Importantly, the algorithm is non-adaptive and so the choice of $S$ cannot depend on $A$, but we are still able to exploit this trick with the following two strategies. Let $\gamma n = |B|$ denote the size of $B$ and note that this implies $|A| \geq (1-2\gamma)n$ since the third cluster is of size at most $B$.

\underline{\emph{Strategy 1:}} Suppose a query $S$ contains exactly one point outside of $A$, i.e. $S \setminus A = \{x\}$ . Then, for $y \notin A$, $\query(S \cup \{y\}) = \query(S)$ iff $x,y$ belong to the same cluster. Thus, we can query $S \cup \{y\}$ for every $y \in U$ to learn the cluster containing $x$. If $S$ is a random set of size $t \approx 1/\gamma$, then the probability that $|S \setminus A| = 1$ is at least $t \cdot \gamma \cdot (1-2\gamma)^{t-1} = \Omega(1)$. Of course, we do not know $\gamma$, but we can try $t = 2^p$ for every $p \leq \log n$ and one of these choices will be within a factor of $2$ from $1/\gamma$. This gives us an $O(n \log n)$ query algorithm since we make $n$ queries per iteration.

\underline{\emph{Strategy 2:}} Suppose $S$ intersects $A$ and contains exactly two points outside of $A$, i.e. $S \setminus A = \{x,y\}$. Then, $\query(\{x,y\}) = \query(S) - 1$ which tells us whether or not $x,y$ belong to the same cluster. If $x,y$ belong to same cluster, add it to a set $E$, and let $G(U \setminus A, E)$ denote a graph on the remaining points with this set of edges. By transitivity, a connected component in this graph corresponds to a subset of one of the remaining two clusters. In particular, if the induced subgraph, $G[B]$, is connected, then we recover $B$. Moreover, if $S$ is a random set of size $t \approx 1/\gamma$, then the probability that two points land in $B$ and the rest land in $A$ is at least ${t \choose 2} \cdot \gamma^2 \cdot (1-2\gamma)^{t-2} = \Omega(1)$. A basic fact from random graph theory says that after $\approx |B|\log|B| \leq \gamma n \ln n$ occurrences of this event, $G[B]$ becomes connected with high probability and so querying $\Omega(\gamma n \log n)$ random $S$ of size $\approx 1/\gamma$ will suffice. Again, we try $t = 2^p$ for every $p \leq \log n$, resulting in a total of $\approx n\log n \sum_{p} 2^{-p} = O(n\log n)$ queries.

In fact, we can combine strategies (1) and (2) as follows to obtain an $O(n \log \log n)$ query algorithm. (More generally, this is how we obtain the algorithm of \Cref{thm:nloglogn-informal} for small $k$.) The main observation is that the query complexity of strategy (2) improves greatly if we assume that $|B|$ is small enough. If we know that $\smash{\gamma \leq \frac{1}{\log n}}$, then we only need to try $t = 2^p \geq \log n$ and so the query complexity becomes $\smash{\approx n\log n\sum_{p\geq \log\log n} 2^{-p} = O(n)}$. On the other hand, if we assume that $\smash{\gamma > \frac{1}{\log n}}$, then in strategy (1) we only need to try $p \leq \log \log n$ yielding a total of $O(n \log \log n)$ queries. Combining these yields the final algorithm.

\paragraph{Bounded query size.} In fact, we can restrict the query size to $s \leq \sqrt{n}$, and still achieve a near-linear query complexity  using strategy (2) above. 
Again, let $A,B$ denote the largest, and second largest clusters, respectively, where $|B| = \gamma n$ and so $|A| \geq (1-2\gamma) n$. If we take a random set $S$ of size $t$ where $\sqrt{1/\gamma} \leq t \leq 1/\gamma$, then the probability that exactly two points land in $B$ and the rest land in $A$ is at least ${t \choose 2} \cdot \gamma^2 \cdot (1-2\gamma)^{t-2} = \Omega(\gamma)$. Recalling the definition of the graph $G$ and the discussion above, after querying $\Omega(n \log n)$ such $S$, the induced subgraph $G[B]$ becomes connected with high probability, thus recovering the clustering. This yields an $O(n \log n \log \log n)$ algorithm since we only need to try every $t \in \{n^{1/2}, n^{1/4}, n^{1/8}, \ldots, 2\}$.

A straightforward extension of this idea lets us generalize to any $s \leq \sqrt{n}$, and obtain an algorithm with query complexity $\widetilde{O}(n^2/s^2)$, achieving the optimal dependency on $s$ for any constant $k$. See \Cref{sec:bounded-2} and \Cref{thm:bounded-2} for the precise statement and a formal proof.

\subsubsection{Combining the Two Ideas for our Main Algorithm} \label{sec:tech-3}

A second useful group testing primitive is the following, which we prove in \Cref{sec:supp-recovery}. 

\begin{restatable}{lemma}{onebitcertone} \label{lem:1-bit-cert-1} Let $v \in \{0,1\}^n$. There is a deterministic, non-adaptive algorithm that makes $2 \lceil \log n \rceil$ $\mathsf{OR}$ queries, and certifies whether $|\supp(v)| = 1$ or $|\supp(v)| \neq 1$. If $|\supp(v)| = 1$, then it outputs $\supp(v)$. \end{restatable}

Again, combining with \Cref{obs:subset->OR} yields the following corollary. 

\begin{restatable}{corollary}{corollaryonebit} \label{cor:1-bit-1} Given a $k$-clustering on $U$ of size $n$ and an element $x \in U$, let $C(x)$ denote the cluster containing $x$. There is a deterministic non-adaptive algorithm, "$\mathsf{FINDREP}$" which takes as input $x$ and a set $R \subseteq U$, makes $4\ceil{\log |R|}$ subset queries, and certifies whether $|R \cap C(x)| = 1$ or $|R \cap C(x)| \neq 1$. If $|R \cap C(x)| = 1$, then it outputs $R \cap C(x)$. \end{restatable}

\begin{proof} Let $R = \{x_1,\ldots,x_{m}\}$ denote a fixed, arbitrary ordering on $R$ and let $v \in \{0,1\}^m$ denote the Boolean vector such that $v_i = 1$ iff $x_i \in C(x)$. By \Cref{obs:subset->OR} we can simulate the algorithm of \Cref{lem:1-bit-cert-1} using only twice as many subset queries. Moreover, clearly if the algorithm outputs $\supp(v)$, then we can return $R \cap C(x)$, and if the algorithm certifies that $|\supp(v)| \neq 1$, then we can certify that $|C(x)| \neq 1$. \end{proof}

\Cref{cor:1-bit-1} is quite versatile; we apply it as a subroutine in our algorithms of \Cref{thm:1-informal,thm:bal-informal,thm:2round-informal}. To demonstrate its utility, suppose we have some set $R \subseteq U$ such that $R$ intersects the cluster $C_i$ on exactly one element, $z$. Then, the cluster $C_i$ can be recovered with $4n \ceil{\log |R|}$ queries by running $\mathsf{FINDREP}(x,R)$ for every $x \in U$. The procedure will identify $z$ as a representative of $x$'s cluster if and only if $x \in C_i$. Consider now the balanced case, where every cluster is of size $n/k$. If one samples a set $R$ containing $k$ uniform random elements of $U$, then for each $i \in [k]$, there is an $\Omega(1)$ probability that $R$ hits $C_i$ exactly once. A straightforward calculation shows that if one samples $q \approx \log k$ such sets $R_1,\ldots,R_q$, then every cluster will be hit exactly once by at least one $R_i$ with high probability. Moreover, if this occurs, then running $\mathsf{FINDREP}(x,R_i)$ for every $x \in U$ and $i \in [q]$ recovers the entire clustering. Since $|R_i| = k$, the total number of queries is $O(n \log^2 k)$, and this yields item (b) of \Cref{thm:bal-informal} (see \Cref{thm:k-bal-2} for the formal statement and proof). 

This idea can be ported to the general (not necessarily balanced) case as follows. Suppose we have some cluster $C$ of size $|C| \approx \frac{n}{2^p}$ for some $p \in [\log n]$. A random $R$ of size $2^p$ will hit $C$ exactly once with probability $\Omega(1)$. Thus, choosing $p$ such sets\footnote{In fact, it suffices to pick $\min(\log k,p)$ sets since there are only $k$ clusters, but this only leads to at best a $\log n$ factor improvement.} $R_1,\ldots,R_p$ will result in every such $C$ being hit exactly once by some $R_i$. Summing over all $p \leq \log n$ yields an algorithm with query complexity $\approx \sum_{p \leq \log n} n \cdot p^2 \approx n \log^3 n$.

To prove our main result \Cref{thm:1-informal} (see \Cref{thm:1} for the formal statement) we combine this idea with strategy (2) from \Cref{sec:tech-2}. The idea is that the strategy described above using the group testing primitive requires less queries when learning larger clusters (when $p$ is small), while strategy (2) of \Cref{sec:tech-2} requires less queries when learning smaller clusters (when $p$ is large). Our algorithm then follows by choosing an appropriate threshold $\tau$ for which we switch strategies once $p > \tau$.


\subsection{Open Questions and Discussion} \label{sec:discussion}

The most obvious question left open by our work is whether it is possible to obtain an $O(n)$ query non-adaptive algorithm, or if it is possible to prove a super-linear lower bound. In fact, we don't know the answer to this question even for \emph{non-adaptive deterministic} algorithms, and the best non-adaptive \emph{randomized} algorithm we know of even for $k=3$ has query complexity $O(n \log \log n)$ (see \Cref{thm:nloglogn}).

\begin{question} Is there a \emph{randomized non-adaptive} algorithm making $O(n)$ queries, even for $k=3$? On the other hand, can one prove a super-linear lower bound for \emph{deterministic non-adaptive} algorithms? \end{question}

Another interesting question is to obtain a linear algorithm using relatively few rounds of adaptivity. The $O(n)$ query algorithm of \cite{CL24} uses $O(\log n)$ rounds, and we have obtained a simple $2$-round algorithm making $O(n \log k)$ queries in \Cref{sec:2round-1}. It is not clear whether our techniques or those of \cite{CL24} can be used to obtain a smooth round-complexity vs. query-complexity trade-off, but we believe this is an interesting direction for future work. Concretely, we leave the following open question.


\begin{question} \label{question:rounds} Is there an $O(n)$ query algorithm using $2$ rounds of adaptivity? If not, what is the minimum round complexity needed to achieve $O(n)$ query complexity? Is there an $O(n)$ query algorithm using $o(\log n)$ rounds?\end{question}

\noindent In fact, the above question is open even in the balanced cluster setting, where the best 2-round algorithm we know of makes $O(n \log \log k)$ queries.

\paragraph{On minimizing query size and recent work since the conference version.} Next, we believe that minimizing the \emph{size} of subset queries is a worthwhile direction to consider. With queries of size $s \leq \sqrt{n}$, we obtain upper and lower bounds of $\widetilde{O}(\frac{n^2}{s} \min(1,\frac{k}{s}))$ and $\Omega(n^2/s^2)$. Since the preliminary conference version of this work this gap has been resolved by \cite{DBLP:conf/colt/BlackMS25} who gave a $\widetilde{O}(n^2/s^2)$ query algorithm, matching the lower bound up to log-factors. We also remark that \cite{DBLP:conf/colt/BlackMS25} give smooth trade-offs with query and round complexity in the bounded query size setting. However, the best known algorithms in this setting incur large $\mathrm{polylog}$-dependencies, and so these results do not speak to the above open questions.



\section{Clustering with Subset Queries of Unbounded Size} \label{sec:unbounded}

In this section we prove our main results regarding algorithms that run without any restriction on the size of the subsets they query. This means that some of the queries made by these algorithms are as large as $\Omega(n)$. In \Cref{sec:unbounded-1} we prove our main result (\Cref{thm:1}), and in \Cref{sec:unbounded-2} we obtain an improved algorithm for the case when $k$ is sufficiently small (\Cref{thm:nloglogn}). 

\subsection{An $O(n \log k\cdot (\log k + \log \log n)^2)$ Algorithm} \label{sec:unbounded-1}

\begin{theorem} \label{thm:1} There is a non-adaptive algorithm for $k$-clustering which for any $\delta > 0$, uses 
\[
O\big(n \cdot \log (k/\delta) \cdot (\log k + \log \log (n/\delta))^2\big)
\]
subset queries and succeeds with probability $1-\delta$. \end{theorem} 

Note that for $\delta = 1/\poly(k)$ and $k = \Omega(\log n)$, this algorithm makes $O(n \log^3 k)$ queries. On the other hand if $k = O(\log n)$, then this algorithm makes $O(n \cdot (\log \log n)^2 \cdot \log k)$ queries. 

\paragraph*{High-level algorithm description.} A full description of the algorithm is given in pseudocode \Alg{1}, which is split into two phases: a "query selection phase", which describes how queries are chosen by the algorithm, and a "reconstruction phase" which describes how the algorithm uses the query responses to determine the clustering. Both phases contain a for-loop iterating over all $p \in \{1,2,\ldots,\log n\}$ where the goal of the algorithm during the $p$'th iteration is to learn all remaining clusters of size at least $\frac{n}{2^p}$. This is accomplished by two different strategies depending on whether $p$ is small or large. We refer the reader to \Cref{sec:tech-2,sec:tech-3} for an intuitive description of the ideas and for a proof of \Cref{cor:1-bit-1} which is used in our algorithm in lines 9-10. \\

\noindent \emph{Learning large clusters:} When $p \leq \log (k^2\log (n/\delta))$, we adapt the strategy used in \Cref{sec:bal-2} for the case of balanced clusters. We will use the following subroutine discussed in \Cref{sec:tech-3}.

\corollaryonebit*


By \Cref{cor:1-bit-1}, if for a cluster $C$, we have a set $R$ such that $|C \cap R| = 1$, then we can learn $C$ using $O(n \log |R|)$ non-adaptive queries by calling $\mathsf{FINDREP}(x,R)$ for every $x \in U$. In particular, for each $x \in C$, $\mathsf{FINDREP}(x,R)$ will return the unique representative of $C$ from $R$, which then enables the recovery of $C$. In line (5) of \Alg{1}, we sample $s$ uniform random subsets $R_1^{(p)},\ldots,R_s^{(p)} \subset U$, each of size $2^p$. For sufficiently large $s$ (see line 5), with high probability every cluster $C$ of size $|C| \in [\frac{n}{2^p},\frac{n}{2^{p-1}}]$ will satisfy $|R^{(p)}_i \cap C| = 1$ for some $i \in [s]$. When this occurs, running the procedure $\mathsf{FINDREP}(x,R^{(p)}_i)$ for each $x \in U$ exactly recovers the cluster $C$ (see lines 9-11 and 29-30). \\


\noindent \emph{Learning small clusters:} When $p > \log (k^2 \log (n/\delta))$, the algorithm queries $O(nk^2 \cdot \frac{\log (n/\delta)}{2^p}) \leq O(n)$ random sets $T$ formed by $2^p/k$ samples from $U$ (see lines 16-19 of \Alg{1}). When all but exactly two points $x,y$ in the query $T$ belong to the already reconstructed clusters (see line 34) we deduce whether or not $x,y$ are in the same cluster (this effectively simulates a pair query on $x,y$). When it is deduced that $x,y$ belong to the same cluster (see line 35) we add the edge $(x,y)$ to the graph $G_p$ in line 36. The connected components of $G_p$ of size at least $n/2^p$ are then guaranteed (with high probability) to be exactly the remaining clusters of this size. The analysis crucially relies on the connectivity threshold for random graphs. See strategy (2) in \Cref{sec:tech-2} for a discussion on the ideas used in this case.

\paragraph*{Analysis.} The following \Cref{lem:main-analysis-1} establishes that after the first $p$ iterations of the algorithm's query selection and reconstruction phases, all clusters of size at least $\frac{n}{2^p}$ have been learned with high probability. This is the main technical component of the proof. After stating the lemma we show it easily implies that \Alg{1} succeeds with probability at least $1-\delta$ by an appropriate union bound. 

\begin{lemma} \label{lem:main-analysis-1} For each $p = 1,\ldots,\log n$, let $\cE_p$ denote the event that all clusters of size at least $\frac{n}{2^p}$ have been successfully recovered immediately following iteration $p$ of \Alg{1}. Then, 
\[
\Pr[\neg \cE_0] \leq \delta/k ~\text{ and } ~\Pr[\neg \cE_p ~|~ \cE_{p-1}] \leq \delta/k~ \text{ for all } p \in \{1,2\ldots,\log n\} \text{.}
\]
\end{lemma}

\paragraph{Proof of \Cref{thm:1}:} Before proving \Cref{lem:main-analysis-1}, we first observe that it immediately implies the correctness of \Alg{1} and thus proves \Cref{thm:1}. For $1 \leq p \leq \log n$, let $I_p = [\frac{n}{2^{p}},\frac{n}{2^{p-1}})$. If there are no clusters $C$ for which $|C| \in I_p$, then trivially $\Pr[\neg \cE_p ~|~ \cE_{p-1}] = 0$, and otherwise $\Pr[\neg \cE_p ~|~ \cE_{p-1}] \leq \delta/k$ by the lemma. Since there are $k$ clusters, clearly there are at most $k$ values of $p$ for which there exists a cluster with size in the interval $I_p$. Using this observation and a union bound, we have
\[
\Pr[\neg \cE_{\log n}] \leq \Pr[\neg \cE_0] + \sum_{p=1}^{\log n} \Pr[\neg \cE_p ~|~ \cE_{p-1}] \leq \delta
\]
which completes the proof of correctness since the algorithm succeeds iff $\cE_{\log n}$ occurs.

\paragraph{Query complexity:} During the $p$'th iteration, where $p < \log(k^2\log (n/\delta))$, the algorithm calls the subroutine from \Cref{cor:1-bit-1} on $n \cdot s = O(n \log (k/\delta))$ pairs $(x,R)$ where $R \subseteq U$ is of size $2^p$. By \Cref{cor:1-bit-1}, each such call to the subroutine uses $O(p)$ queries. Thus, the total number of queries made during the first $\log(k^2\log \frac{n}{\delta})$ iterations is
\[
O(n \log(k/\delta)  ) \cdot \left( \sum_{p = 1}^{\log (k^2 \log (n/\delta))} p \right) = O\big(n \cdot \log (k/\delta) \cdot (\log k + \log \log (n/\delta))^2\big) \text{.}
\]
The total number of queries made during iterations $p > \log(k^2\log (n/\delta))$ is at most
\[
O(nk^2 \log (n/\delta)) \sum_{p > \log(k^2 \log (n/\delta))} 2^{-p} = O(n)
\]
since $k \leq n$. This completes the proof of \Cref{thm:1}. \qed \\

\begin{algorithm}
\caption{Main Algorithm for Unrestricted Query Size} \label{alg:1}
\textbf{Input:} Subset query access to a hidden partition $C_1 \sqcup \cdots \sqcup C_k = U$ of $|U| = n$ points\;
\emph{(Query Selection Phase)} \\
\For{$p = 1,2,\ldots,\log n$} {
\If{$p \leq \log (k^2 \log (n/\delta))$} { 
Sample $s = 5\log(k^2/\delta))$ sets $R_1^{(p)},\ldots,R_{s}^{(p)}$ each formed by $2^p$ uniform samples from $U$\;
Let $R^{(p)} = \bigcup_{\ell \in [s]} R^{(p)}_{\ell}$\;
Initialize $E_p \gets \emptyset$ and consider the bipartite graph $G_p(U,R^{(p)},E_p)$\;
\For{$x \in U$, $i \in [s]$} {
        Run the non-adaptive algorithm $\mathsf{FINDREP}$ from \Cref{cor:1-bit-1} on $(x,R_i)$\;
        \If{$\mathsf{FINDREP}(x,R_i)$ outputs $y \in R_i$} {
            Add edge $(x,y)$ to $E_p$\;
        }
}
}
\If{$p > \log(k^2 \log (n/\delta))$} {
Initialize $Q_p \gets \emptyset$\;
\textbf{Repeat} $\frac{40 nk^2 \ln (3nk^2/\delta)}{2^p}$ times\;
$\longrightarrow$ Sample $T \subseteq U$ formed by $\frac{2^p}{k}$ independent uniform samples from $U$\;
$\longrightarrow$ \textbf{Query} $T$ and add it to $Q_p$\;
}
}
\emph{(Reconstruction Phase)} \\
Initialize learned cluster set $\cC_1 \gets \emptyset$\;
\For{$p = 1,2,\ldots,\log n$} {
Let $\cC_p$ denote the collection of clusters reconstructed before iteration $p$\;
Let $\cR_p = \bigcup_{C \in \cC_p} C$ denote the points belonging to these clusters\;
Initialize $\cC_{p+1} \gets \cC_p$\;
\If{$p \leq \log(k^2 \log (n/\delta))$} {
    Consider the bipartite graph $G_p(U,R^{(p)},E_p)$ constructed in lines 5-13\;
    Let $\widetilde{C}_1,\ldots,\widetilde{C}_{\ell} \subseteq U$ denote the connected components of $G_p$ (taking only vertices on the left) with size at least $\frac{n}{2^p}$\;
    Add $\widetilde{C}_1,\ldots,\widetilde{C}_{\ell}$ to $\cC_{p+1}$\;
}
\If{$p > \log(k^2 \log (n/\delta))$} {
Let $Q_p' = \{T \setminus \cR_p \colon T \in Q_p \text{ and } |T \setminus \cR_p| = 2\}$. Since each $T \in Q_p$ is a uniform random set, the elements of $Q_p'$ are uniform random pairs in $U \setminus \cR_p$\;
Let $Q_p'' = \{\{x,y\} \in Q_p' \colon \query(\{x,y\} = 1)\}$ denote the set of pairs in $Q_p'$ where both points lie in the same cluster. This set can be computed since $\query(T \setminus \cR_p) = \query(T) - \query(T \cap \cR_p)$ and $\query(T \cap \cR_p)$ is known since at this point we have reconstructed the clustering on $\cR_p$\;
Let $G_p$ denote the graph with vertex set $U \setminus \cR_p$ and edge set $Q_p''$\;
Let $\widetilde{C}_1,\ldots,\widetilde{C}_{\ell}$ denote the connected components of $G_p$ with size at least $\frac{n}{2^{p}}$\; 
Add $\widetilde{C}_1,\ldots,\widetilde{C}_{\ell}$ to $\cC_{p+1}$\;
}
}
\textbf{Output} clustering $\cC_{\log n + 1}$
\end{algorithm}

We now prove the main \Cref{lem:main-analysis-1}.
\begin{proof} \emph{of \Cref{lem:main-analysis-1}.} Let $\cC_p$ denote the set of clusters recovered before phase $p$ and let $\cR_p = \bigcup_{C \in \cC_p} C$. When $p = 1$, both of these sets are empty. We will consider two cases depending on the value of $p$. \\

\noindent \textbf{Case 1:} $1 \leq p \leq \log (k^2 \log (n/\delta))$. Let $C_1,\ldots,C_t$ denote the clusters with size $|C_j| \in [\frac{n}{2^{p}},\frac{n}{2^{p-1}}]$. Note that $t \leq \min(k, 2^p)$. By \Cref{cor:1-bit-1} and the definition of \Alg{1} (see lines 5-13 and lines 29-31) if some $R^{(p)}_i$ sampled in line (5) satisfies $|R^{(p)}_i \cap C_j| = 1$, then $C_j$ appears as a connected component in lines (30-31) during the $p$'th iteration of the reconstruction phase, and is thus correctly recovered. Let $\cA_{i,j}$ denote the (good) event that $|R^{(p)}_i \cap C_j| = 1$. Then, for fixed $i \in [s]$ and $j \in [t]$, we have

\[
\Pr_{R^{(p)}_i}\left[\cA_{i,j}\right] = |R^{(p)}_i| \cdot \frac{|C_j|}{n} \cdot \left(1-\frac{|C_j|}{n}\right)^{|R^{(p)}_i|-1} \geq 2^p \cdot \frac{1}{2^p} \cdot \left(1-\frac{1}{2^{p-1}}\right)^{2^p} \geq \frac{1}{e^2} \text{.}
\]
We need for every $j \in [t]$, that $\cA_{i,j}$ occurs for some $i \in [s]$. Fix $j \in [t]$. Since $R^{(p)}_1,\ldots,R^{(p)}_{s}$ are independent, we have
\[
\Pr_{R^{(p)}_1,\ldots,R^{(p)}_{s}}\left[\bigwedge_{i \in [s]} \neg \cA_{i,j}\right] \leq (1-1/e^2)^{5\log(k^2/\delta)} \leq 2^{-\log(k^2/\delta)} = \frac{\delta}{k^2} \text{.}
\]
Finally, since $t \leq k$, by a union bound we have
\[
\Pr_{R^{(p)}_1,\ldots,R^{(p)}_{s}}\left[\bigvee_{j \in [t]} \bigwedge_{i \in [s]} \neg \cA_{i,j}\right] \leq \frac{\delta}{k}
\]
which completes the proof. \\

\noindent \textbf{Case 2: } $p > \log (k^2 \log (n/\delta))$. Let $C$ denote some cluster with size $|C| \in [\frac{n}{2^{p}},\frac{n}{2^{p-1}})$. Note that we are conditioning on all clusters of size greater than $\frac{n}{2^{p-1}}$ having already been learned (i.e. every such cluster belongs to $\cC_{p}$) and so $|U \setminus \cR_p| \leq k \cdot \frac{n}{2^{p-1}} = \frac{nk}{2^{p-1}}$. Recall from lines (34-36) the definition of $Q_p''$ and recall that $G_p$ is the graph with vertex set $U \setminus \cR_p$ and edge set $Q_p''$. We need to argue that the induced subgraph $G_p[C]$ is connected, and thus $C$ is successfully recovered in lines (36-38), with probability at least $1-\delta/k^2$. Once this is established, the lemma follows by a union bound since there are at most $k$ such clusters. We rely on the following standard fact from the theory of random graphs. For completeness, we give a proof in \Cref{sec:random-graph}.

\begin{restatable}{fact}{randomconnected} \label{fact:random-connected} Let $\smash{G(N,p)}$ denote an Erd\H{o}s-R\'{e}nyi random graph. That is, $G$ contains $N$ vertices and each pair is an edge independently with probability $p$. If $\smash{p \geq 1-\left(\alpha/3N\right)^{2/N}}$, then $G(N,p)$ is connected with probability at least $1-\alpha$. \end{restatable}

Recall $Q_p$ defined in line (16) which is a collection of sets $T \subseteq U$ each of which is formed by $\frac{2^p}{k}$ independent uniform samples from $U$. Consider any $x,y \in C$ and observe that 
\begin{align} 
    \Pr_{T \colon |T| = 2^p/k}[T \setminus \cR_p = \{x,y\}] = {2^p/k \choose 2} \cdot \frac{1}{n^2} \cdot \left(\frac{|\cR_p|}{n}\right)^{\frac{2^p}{k}-2} \geq \frac{2^{2p}}{3k^2n^2}  \left(1-\frac{k}{2^{p-1}}\right)^{\frac{2^{p}}{k}} \geq \frac{2^{2p}}{25k^2n^2} \text{.} \nonumber
\end{align} 
Recall that the algorithm queries $|Q_p| = \frac{50 \cdot nk^2 \ln (3nk^2/\delta)}{2^p}$ random sets $T$ of size $2^p$. Thus,
\begin{align}
    \Pr_{Q_p}\left[(x,y) \in E(G_p[C])\right] &= \Pr_{Q_p}\left[\{x,y\} \in Q_p''\right] = \Pr_{Q_p}\left[\exists T \in Q_p \colon T \setminus \cR_p = \{x,y\}\right] \nonumber \\
    &\geq 1-\left(1-\frac{2^{2p}}{25k^2n^2}\right)^{50\frac{n}{2^p} \cdot k^2\ln (3 nk^2/\delta)} \nonumber \\
    &\geq 1-\exp\left(-\frac{2^p}{n} \cdot 2 \ln (3nk^2/\delta)\right) \nonumber
\end{align}
and using $|C| \geq \frac{n}{2^{p}}$ and $|C| \leq n$, we obtain
\begin{align}
   \Pr_{Q_p}\left[(x,y) \in E(G_p[C])\right] &\geq 1-\exp\left(-\frac{2\ln (3nk^2/\delta)}{|C|}\right) \nonumber \\
   &\geq 1-\exp\left(-\frac{2\ln (3 k^2|C|/\delta)}{|C|}\right) = 1-\left(\frac{\delta}{3k^2|C|}\right)^{\frac{2}{|C|}} \text{.} \nonumber 
\end{align}
Thus, $(x,y)$ is an edge in $G_p[C]$ with probability at least $1-\left(\frac{\delta}{3k^2|C|}\right)^{\frac{2}{|C|}}$ and so by \Cref{fact:random-connected} $G_p[C]$ is connected with probability at least $1-\delta/k^2$, as claimed. This completes the proof of \Cref{lem:main-analysis-1}. \end{proof}

\subsection{An Improved Algorithm for Small $k$} \label{sec:unbounded-2}

\begin{theorem} \label{thm:nloglogn} There is a non-adaptive algorithm for $k$-clustering which for any $\delta > 0$, uses 
\[
O(n\log\log (n/\delta) \cdot k\log (k/\delta))
\]
subset queries and succeeds with probability at least $1-\delta$.
\end{theorem}

Note that when $k = O(1)$ and $\delta = 1/\poly(n)$, this algorithm makes $O(n \log \log n)$ queries, improving on \Cref{thm:1} by a $\log \log n$ factor.

\paragraph*{Algorithm.}
A full description of the algorithm is given in pseudocode \Alg{nloglogn}, which is split into two phases: a "query selection phase", which describes how queries are chosen by the algorithm, and a "reconstruction phase" which describes how the algorithm uses the query responses to determine the clustering. Both phases contain a for-loop iterating over all $p \in \{0,1,\ldots,\log n\}$ where the goal of the algorithm during the $p$'th iteration is to learn all remaining clusters of size at least $\frac{n}{2k \cdot 2^p}$. This is accomplished by two different strategies depending on whether $p$ is small or large. See \Cref{sec:tech-2} for an informal description of the algorithm for the case of $k=3$.

When $p \leq \log \log (n/\delta)$, the algorithm samples $O(k \log (k/\delta))$ random sets $T$ formed by $2^p$ samples from $U$ and makes a query on $T$ and $T \cup \{x\}$ for every $x \in U$ (see lines 5-9 of \Alg{nloglogn}). 
Let $\cR_p$ be the union of all clusters reconstructed before phase $p$ (i.e., clusters of size at least $\frac{n}{2k\cdot2^{p-1}}$). 
If such a $T$ contains exactly one point $z \in T \setminus \cR_p$ belonging to an unrecovered cluster, then we can use these queries to learn the cluster containing $z$ (see lines 24-28 of \Alg{nloglogn}), since for $x \in U \setminus \cR_p$, $\query(T) = \query(T \cup \{x\})$ if and only if $x,z$ belong to the same cluster. Moreover, we show that this occurs with probability $\Omega(1)$ and repeat this $O(k \log (k/\delta))$ times to ensure that all clusters with size in the interval $[\frac{n}{2k \cdot 2^{p}}, \frac{n}{2k \cdot 2^{p-1}})$ are learned with probability $1-\delta$. 

When $p > \log \log (n/\delta)$, the algorithm queries $O(nk \cdot \frac{\log (n/\delta)}{2^p})$ random sets $T$ again formed by $2^p$ samples from $U$ (see lines 11-14 of \Alg{nloglogn}). If $T$ contains exactly two points $x,y \in T \setminus \cR_p$ belonging to unrecovered clusters, then we can use the fact that we already know the clustering on $\cR_p$ to tell whether or not $x,y$ belong to the same cluster or not, i.e. we can compute $\query(\{x,y\}) \in \{1,2\}$ from $\query(T)$. We then consider the set of all such pairs where $\query(\{x,y\}) = 1$ (this is $Q_p''$ defined in line 34) and consider the graph $G$ with this edge set, and vertex set $U \setminus \cR_p$, the set of points whose cluster hasn't yet been determined. If two points belong to the same connected component in this graph, then they belong to the same cluster. Thus, the analysis for this iteration boils down to showing that with high probability, the induced subgraph $G[C]$ will be connected for every $C$ where $|C| \in [\frac{n}{2k \cdot 2^{p}}, \frac{n}{2k \cdot 2^{p-1}})$. This is accomplished by applying a basic fact from the theory of random graphs, namely \Cref{fact:random-connected}.

\paragraph*{Analysis.} The following \Cref{lem:main-analysis} establishes that after the first $p$ iterations of the algorithm's query selection and reconstruction phases, all clusters of size at least $\frac{n}{2k \cdot 2^p}$ have been learned with high probability. This is the main technical component of the proof. After stating the lemma we show it easily implies that \Alg{nloglogn} succeeds with probability at least $1-\delta$ by an appropriate union bound.

\begin{lemma} \label{lem:main-analysis} For each $p = 0,1,\ldots,\log n$, let $\cE_p$ denote the event that all clusters of size at least $\frac{n}{2k \cdot 2^p}$ have been successfully recovered immediately following iteration $p$ of \Alg{nloglogn}. Then, 
\[
\Pr[\neg \cE_0] \leq \delta/k ~ \text{ and } ~\Pr[\neg \cE_p ~|~ \cE_{p-1}] \leq \delta/k~ \text{ for all } p \in \{1,2\ldots,\log n\} \text{.}
\]
\end{lemma}
\paragraph{Proof of \Cref{thm:nloglogn}:} Before proving \Cref{lem:main-analysis}, we first observe that it immediately implies the correctness of \Alg{nloglogn} and thus proves \Cref{thm:nloglogn}. Let $I_0 = (\frac{n}{2k},n]$ and for $1 \leq p \leq \log n$, let $I_p = [\frac{n}{2k \cdot 2^{p}},\frac{n}{2k \cdot 2^{p-1}})$. If there are no clusters $C$ for which $|C| \in I_p$, then trivially $\Pr[\neg \cE_p ~|~ \cE_{p-1}] = 0$, and otherwise $\Pr[\neg \cE_p ~|~ \cE_{p-1}] \leq \delta/k$ by the lemma. Since there are $k$ clusters, clearly there are at most $k$ values of $p$ for which there exists a cluster with size in the interval $I_p$. Using this observation and a union bound, we have
\[
\Pr[\neg \cE_{\log n}] \leq \Pr[\neg \cE_0] + \sum_{p=1}^{\log n} \Pr[\neg \cE_p ~|~ \cE_{p-1}] \leq \delta
\]
which completes the proof of correctness since the algorithm succeeds iff $\cE_{\log n}$ occurs. 

\paragraph{Query complexity:} During iterations $p < \log\log (n/\delta)$ the algorithm makes at most $O(n \log \log (n/\delta) \cdot k \log (k/\delta))$ queries. During iterations $p > \log\log (n/\delta)$, it makes at most $O(nk \log (n/\delta)) \sum_{p > \log\log (n/\delta)} 2^{-p} = O(nk)$ queries since $k \leq n$. This completes the proof of \Cref{thm:nloglogn}. \qed \smallskip

\begin{algorithm}
\caption{Improved Algorithm for Small $k$ with Unrestricted Query Size} \label{alg:nloglogn}
\textbf{Input:} Subset query access to a hidden partition $C_1 \sqcup \cdots \sqcup C_k = U$ of $|U| = n$ points\;
\emph{(Query Selection Phase)} \\
\For{$p = 0,1,\ldots,\log n$} {
Initialize $Q_p \gets \emptyset$\;
\If{$p \leq \log \log (n/\delta)$} { 
\textbf{Repeat} $2e k \ln (k^2/\delta)$ times\;
$\longrightarrow$ Sample $T \subseteq U$ formed by $2^p$ independent uniform samples from $U$\; 
$\longrightarrow$ \textbf{Query} $T$ and $T \cup \{x\}$ for all $x \in U$\; 
$\longrightarrow$ Add $T$ to $Q_p$\;
}
\If{$p > \log\log (n/\delta)$} {
\textbf{Repeat} $\frac{40 nk \ln (3nk^2/\delta)}{2^p}$ times\;
$\longrightarrow$ Sample $T \subseteq U$ formed by $2^p$ independent uniform samples from $U$\;
$\longrightarrow$ \textbf{Query} $T$ and add it to $Q_p$\;
}
}
\emph{(Reconstruction Phase)} \\
Initialize learned cluster set $\cC_0 \gets \emptyset$\;
\For{$p = 0,1,\ldots,\log n$} {
Let $\cC_p$ denote the collection of clusters reconstructed before iteration $p$\;
Let $\cR_p = \bigcup_{C \in \cC_p} C$ denote the points belonging to these clusters\;
Initialize $\cC_{p+1} \gets \cC_p$\;
\If{$p \leq \log\log (n/\delta)$} {
\For{$T \in Q_p$} { 
\If{$|T \setminus \cR_p| = 1$} {
    Let $z$ denote the unique point in $T \setminus \cR_p$\;
    If $x \in U \setminus \cR_p$, then $\query(T) = \query(T \cup \{x\})$ iff $x,z$ are in the same cluster\;
    Thus, we add $\{x \in U \setminus \cR_p \colon \query(T) = \query(T \cup \{x\})\}$ to $\cC_{p+1}$\;
    
    }
}
}
\If{$p > \log\log (n/\delta)$} {
Let $Q_p' = \{T \setminus \cR_p \colon T \in Q_p \text{ and } |T \setminus \cR_p| = 2\}$. Since each $T \in Q_p$ is a uniform random set, the elements of $Q_p'$ are uniform random pairs in $U \setminus \cR_p$\;
Let $Q_p'' = \{\{x,y\} \in Q_p' \colon \query(\{x,y\} = 1)\}$ denote the set of pairs in $Q_p'$ where both points lie in the same cluster. This set can be computed since $\query(T \setminus \cR_p) = \query(T) - \query(T \cap \cR_p)$ and $\query(T \cap \cR_p)$ is known since at this point we have reconstructed the clustering on $\cR_p$\;
Let $G_p$ denote the graph with vertex set $U \setminus \cR_p$ and edge set $Q_p''$\;
Let $\widetilde{C}_1,\ldots,\widetilde{C}_{\ell}$ denote the connected components of $G_p$ with size at least $\frac{n}{2k \cdot 2^{p}}$\; 
Add $\widetilde{C}_1,\ldots,\widetilde{C}_{\ell}$ to $\cC_{p+1}$\;
}
}
\textbf{Output} clustering $\cC_{\log n + 1}$
\end{algorithm}

\begin{proof} \emph{of \Cref{lem:main-analysis}.} Let $\cC_p$ denote the set of clusters recovered before phase $p$ and let $\cR_p = \bigcup_{C \in \cC_p} C$. When $p = 0$, both of these sets are empty. We will consider three cases depending on the value of $p$. \smallskip

\noindent \textbf{Case 1:} $p = 0$. Let $C$ denote some cluster of size $|C| \geq \frac{n}{2k}$. Note that in this iteration the sets $T$ sampled by the algorithm in line (7) are singletons. We need to argue that one of these singletons will land in $C$, and thus $C$ is recovered in line (28), with probability at least $1-\delta/k^2$. Since there are at most $k$ clusters, applying a union bound completes the proof in this case.

A uniform random element lands in $C$ with probability at least $1/2k$ and so this fails to occur for all $|Q_0| \geq 2k\ln (k^2/\delta)$ samples with probability at most $(1-\frac{1}{2k})^{2k \ln (k^2/\delta)} \leq \exp(- \ln (k^2/\delta)) = \delta/k^2$, as claimed. \smallskip

\noindent \textbf{Case 2:} $1 \leq p \leq \log \log (n/\delta)$. Let $C$ denote some cluster with size $|C| \in [\frac{n}{2k \cdot 2^{p}},\frac{n}{2k \cdot 2^{p-1}})$. Note that we are conditioning on the event that every cluster of size $\geq \frac{n}{2k \cdot 2^{p-1}}$ has already been successfully recovered after iteration $p-1$. Thus, the number of elements belonging to unrecovered clusters is $|U \setminus \cR_p| \leq k \cdot \frac{n}{2k \cdot 2^{p-1}} = \frac{n}{2^p}$. We need to argue that the set $Q_p$ will contain some $T$ sampled in line (7) such that $T \setminus \cR_p = \{z\}$ where $z \in C$, and thus $C$ is successfully recovered in line (28), with probability at least $1-\delta/k^2$. Once this is established, the lemma again follows by a union bound. We have
\[
\Pr_{T \colon |T| = 2^p}[|T \setminus \cR_p| = 1 \text{ and } T \setminus \cR_p \subseteq C] = |T| \cdot \frac{|C|}{n} \cdot \left(\frac{|\cR_p|}{n}\right)^{|T|-1} \geq \frac{2^p}{k \cdot 2^{p+1}} \left(1-\frac{1}{2^{p}}\right)^{2^{p}} \geq \frac{1}{2e k}
\]
and so the probability that this occurs for some $T \in Q_p$ is at least $1 - (1-\frac{1}{2ek})^{2ek \ln (k^2/\delta)} \geq 1 - \delta/k^2$, as claimed. 

\noindent \textbf{Case 3: } $p > \log \log (n/\delta)$. Let $C$ denote some cluster with size $|C| \in [\frac{n}{2k \cdot 2^{p}},\frac{n}{2k \cdot 2^{p-1}})$. Note that $|U \setminus \cR_p| \leq k \cdot \frac{n}{2k \cdot 2^{p-1}} = \frac{n}{2^{p}}$. Recall from lines (34-35) the definition of $Q_p''$ and recall that $G_p$ is the graph with vertex set $U \setminus \cR_p$ and edge set $Q_p''$. We need to argue that the induced subgraph $G_p[C]$ is connected, and thus $C$ is successfully recovered in lines (36-37), with probability at least $1-\delta/k^2$. Once this is established, the lemma again follows by a union bound. We again rely on the following standard fact. 

\randomconnected*

Consider any $x,y \in C$ and observe that 
\begin{align} 
    \Pr_{T \colon |T| = 2^p}[T \setminus \cR_p = \{x,y\}] = {2^p \choose 2} \cdot \frac{1}{n^2} \cdot \left(\frac{|\cR_p|}{n}\right)^{2^p-2} \geq \frac{2^{2p}}{3n^2}  \left(1-\frac{1}{2^{p}}\right)^{2^p} \geq \frac{2^{2p}}{10n^2} \text{.} \nonumber
\end{align} 
Recall that the algorithm queries $|Q_p| = \frac{40 \cdot nk \ln (3nk^2/\delta)}{2^p}$ random sets $T$ of size $2^p$. Thus,
\begin{align}
    \Pr_{Q_p}\left[(x,y) \in E(G_p[C])\right] &= \Pr_{Q_p}\left[\{x,y\} \in Q_p''\right] = \Pr_{Q_p}\left[\exists T \in Q_p \colon T \setminus \cR_p = \{x,y\}\right] \nonumber \\
    &\geq 1-\left(1-\frac{2^{2p}}{10n^2}\right)^{40\frac{n}{2^p} \cdot k\ln (3 nk^2/\delta)} \nonumber \\
    &\geq 1-\exp\left(-\frac{2^p}{n} \cdot 4k\ln (3nk^2/\delta)\right) \nonumber
\end{align}
and using $|C| \geq \frac{n}{2k \cdot 2^{p}}$ and $|C| \leq n$, we obtain
\begin{align}
   \Pr_{Q_p}\left[(x,y) \in E(G_p[C])\right] &\geq 1-\exp\left(-\frac{2\ln (3 nk^2/\delta)}{|C|}\right) \nonumber \\
   &\geq 1-\exp\left(-\frac{2\ln (3 k^2|C| /\delta)}{|C|}\right) = 1-\left(\frac{\delta}{3k^2|C|}\right)^{\frac{2}{|C|}} \text{.} \nonumber 
\end{align}
Thus, $(x,y)$ is an edge in $G_p[C]$ with probability at least $1-\left(\frac{\delta}{3k^2|C|}\right)^{\frac{2}{|C|}}$ and so by \Cref{fact:random-connected} $G_p[C]$ is connected with probability at least $1-\delta/k^2$, as claimed. This completes the proof of \Cref{lem:main-analysis}. \end{proof}

\section{Clustering with Subset Queries of Size at Most $s$} \label{sec:bounded}

In addition to query complexity, we also view the \emph{size} of queries as an interesting computational resource to consider. In this section, we design algorithms that use subset queries of size only at most some fixed $s$, where $2 \leq s \leq n$. We obtain algorithms\footnote{We remark that since the preliminary conference version of this work these results have been subsumed by \cite{DBLP:conf/colt/BlackMS25}, who discovered a nearly-optimal $\widetilde{O}(n^2/s^2)$ query algorithm for $2 \leq s \leq \sqrt{n}$.} making $\widetilde{O}(n^2/s)$ queries in \Cref{sec:bounded-1} and $\widetilde{O}(n^2k/s^2)$ queries (for $s \leq \sqrt{n}$) in \Cref{sec:bounded-2}.

\subsection{An $\widetilde{O}(n^2/s)$ Query Algorithm via Group Testing} \label{sec:bounded-1}

In this section we describe the following alternative non-adaptive algorithm using a simple group testing primitive. We refer the reader to \Cref{sec:tech-2} for an informal description of the idea. 

\begin{theorem} \label{thm:bounded-1} For every $s \in [2,n]$, there is a non-adaptive algorithm for $k$-clustering which for any $\delta > 0$, uses 
\[
O\left(n \log (n/\delta) \log (k/\delta) \cdot \left(\frac{n}{s} + \log s\right)\right)
\] 
subset queries of size at most $s$ and succeeds with probability $1-\delta$. In particular, letting $\delta = 1/\poly(k)$, for unbounded query size the algorithm makes $O(n \log^2 n \log k)$ queries, and when $s = O(n / \log n)$ the algorithm makes $O(\frac{n^2}{s}\log n \log k)$ queries. \end{theorem}

\begin{proof} The main subroutine in our algorithm is a procedure for recovering the support of a Boolean vector via $\mathsf{OR}$ queries. Given a vector $v \in \{0,1\}^n$, an $\mathsf{OR}$ query on a set $S \subseteq [n]$ returns $\mathsf{OR}_S(v) = \bigvee_{i \in S} v_i$, i.e. it returns $1$ iff $v$ has a $1$-valued coordinate in $S$. The main group testing primitive we use is due to the following lemma, which we prove in \Cref{sec:supp-recovery}.

\tER*

Moreover, we make the following observation which allows us to obtain a group testing primitive which queries only sets of size at most $s'$.

\begin{observation} \label{obs:s-blocks} A single $\mathsf{OR}$ query on a set $S$ can be simulated by $\ceil{|S|/s'}$ queries of size at most $s'$. \end{observation}

\begin{proof} Let $v \in \{0,1\}$ and $S \subseteq [n]$. Partition $S$ arbitrarily into $\ell = \ceil{|S|/s'}$ sets $S_1,\ldots,S_{\ell}$ of size at most $s'$. Then, $\mathsf{OR}_S(v) = \mathsf{OR}_{S_1}(v) \vee \cdots \vee \mathsf{OR}_{S_{\ell}}(v)$. Thus, the $\mathsf{OR}$ query on $S$ can be simulated by $\mathsf{OR}$ queries on $S_1,\ldots,S_{\ell}$. \end{proof}

Combining this observation with \Cref{lem:t-ER} gives us the following group testing primitive using bounded $\mathsf{OR}$ queries. 

\begin{restatable}{lemma}{tERbounded} \label{lem:t-ER-bounded-1} Let $v\in \{0,1\}^n$ and $s',t \geq 1$ be positive integers where $s' \leq \ceil{n/t}$. There is a non-adaptive algorithm that makes $O(\frac{n}{s'} \log (n/\alpha))$ $\mathsf{OR}$ queries on subsets of size at most $s'$, and if $|\supp(v)| \le t$, returns $\supp(v)$ with probability $1-\alpha$, and otherwise certifies that $|\supp(v)| > t$ with probability $1$. \end{restatable}

\begin{proof} \Cref{lem:t-ER} gives an algorithm making $O(t \log (n/\alpha))$ queries on subsets of size $\ceil{n/t}$ and by \Cref{obs:s-blocks}, each such query can be simulated by $O(n/ts')$ queries of size at most $s'$. Thus, we obtain an algorithm making 
\[
O\left(\frac{n}{ts'} \cdot t \log (n/\alpha)\right) = O\left(\frac{n}{s'} \log (n/\alpha)\right)
\]
queries of size at most $s'$. \end{proof}

The utility of the lemma for our problem is illuminated by the following observation.

\subsetOR*

Combining \Cref{lem:t-ER-bounded-1} and \Cref{obs:subset->OR} immediately yields the following corollary which serves as the main subroutine in our algorithm.

\begin{corollary} \label{cor:cluster-recovery} Let $x \in C$ and $r \geq 2, t \geq 1$ be positive integers where $s' \leq \frac{n}{t}$. Let $\alpha \in (0,1)$. There is a non-adaptive algorithm, $\mathsf{FIND}$-$\mathsf{CLUST}(x,t,s',\alpha)$ that makes $O(\frac{n}{s'} \log (n/\alpha))$ \emph{subset}-queries on sets of size at most $s'$, and if $|C| \leq t$, returns $C$ with probability $1-\alpha$, and otherwise certifies that $|C| > t$ with probability $1$. 
\end{corollary}

\paragraph*{Algorithm.} The pseudocode for the algorithm is given in \Alg{general-k}. Given $x$, we use the notation $C(x)$ to denote the cluster containing $x$.  The idea is to draw random points $x \in U$  (line 5) and then use the procedure $\mathsf{FIND}$-$\mathsf{CLUST}(x,\cdot,\cdot,\cdot)$ of \Cref{cor:cluster-recovery} as a subroutine to try to learn $C(x)$ (line 7). By the corollary, this will succeed with high probability in recovering $C(x)$ as long as $t$ is set to something larger than $|C(x)|$. Note that the query complexity of this subroutine depends\footnote{For intuition, if the subroutine is called with $s' = O(n/t)$, then \Cref{cor:cluster-recovery} makes $O(t \log (n/\alpha))$ queries.} on $t$. If a cluster $C$ is small, then $\Pr[x \in C]$ is small, but we can call the subroutine with small $t$, while if $C(x)$ is large, then $\Pr[x \in C]$ is reasonably large, though we will need to call the subroutine with larger $t$. Concretely, the algorithm iterates over every $p \in \{1,\ldots,\log n\}$ (line 3), and in iteration $p$ the goal is to learn every cluster $C$ with $|C| \in [\frac{n}{2^p},\frac{n}{2^{p-1}}]$. To accomplish this, we sample $\Theta(2^p \log (k/\delta))$ random points $x \in U$ (line 4-5) and for each one, call the subroutine with $t = \frac{n}{2^{p-1}}$ (line 6-7), which is an upper bound on the sizes of the clusters we are trying to learn. Note that we always invoke the corollary with query size $s' = \min(s,2^{p-1}) \leq s$, enforcing the query size bound stated in \Cref{thm:bounded-1}.

\begin{algorithm}
\caption{Bounded Query Algorithm using Group Testing}\label{alg:general-k}
\textbf{Input:} Subset query access to a hidden partition $C_1 \sqcup \cdots \sqcup C_k = U$ of $|U| = n$ points\;
Initialize hypothesis clustering $\cC \gets \emptyset$\;
\For{$p = 1,\ldots, \log n$}{
    \textbf{Repeat} $2^p \ln (2k/\delta)$ times: \\ 
    $\longrightarrow$ Sample $x \in U$ uniformly at random\;
    $\longrightarrow$ Set $t \gets \frac{n}{2^{p-1}}$, $s' \gets \min(s,2^{p-1})$, and $\alpha \gets \frac{\delta}{2k}$\;
    $\longrightarrow$ Run the procedure $\mathsf{FIND}$-$\mathsf{CLUST}(x,t,s',\alpha)$ of \Cref{cor:cluster-recovery}. (This outputs $C(x)$, the cluster containing $x$, with probability at least $1-\alpha$ if $|C(x)| \leq t$.)\;
    $\longrightarrow$ If the procedure returns a set $C$, then set $\cC \gets \cC \cup \{C\}$. Otherwise, continue\;
}
\textbf{Output} the clustering $\cC$.
\end{algorithm}

\paragraph{Query complexity:} When $2^{p-1} \geq s$, the number of queries made in line (7) during the $p$'th iteration is $O(\frac{n}{s}\log (n/\delta))$. When $2^{p-1} < s$, the number of queries is $O(\frac{n}{2^p} \log (n/\delta))$. Therefore, the total number of queries made is at most 
\[
O(\log (k/\delta)) \left( \sum_{p \colon 2^{p-1} < s} O\Big(2^p \cdot \frac{n}{2^p} \log (n/\delta)\Big) + \sum_{p \colon s \leq 2^{p-1}} O\Big(2^p \cdot \frac{n}{s}\log (n/\delta)\Big) \right) \text{.}
\]
The first sum is bounded by $O(n \log (n/\delta) \log s)$. The second sum is bounded by $O(\frac{n^2}{s} \log (n/\delta))$ since $\sum_{p = 1}^{\log n} 2^p = O(n)$. 


\paragraph{Correctness:} Consider any cluster $C$ and let $p \in \{1,\ldots,\log n\}$ be such that $\frac{n}{2^{p}} \leq |C| \leq \frac{n}{2^{p-1}}$. Let $\cE_C$ denote the event that some element $x \in C$ is sampled in line (5) during iteration $p$. Let $\cR_C$ denote the event that $C \in \cC$ when the algorithm terminates. Observe that by \Cref{cor:cluster-recovery}, $\Pr[\cR_C ~|~ \cE_C] \geq 1-\alpha = 1-\delta/2k$. Moreover, using our lower bound on $C$ we have
\[
\Pr[\neg \cE_C] \leq \left(1-\frac{|C|}{n}\right)^{2^p \ln (2k/\delta)} \leq \left(1-\frac{1}{2^p}\right)^{2^p \ln (2k/\delta)} \leq \frac{\delta}{2k} \text{.}
\]
Thus, $\Pr[\neg \cR_C] \leq \Pr[\neg \cE_C] + \Pr[\neg \cR_C ~|~ \cE_C] \leq \delta/k$ by a union bound. Taking another union bound over all $k$ clusters  completes the proof. \end{proof}

\subsection{An $\widetilde{O}(n^2 k/s^2)$ Algorithm} \label{sec:bounded-2}

In this section we present an algorithm which for constant $k$ matches the lower bound of \Cref{thm:3-2-NA-LB} up to a factor of $\log n \log \log n$. In particular, this algorithm has the optimal dependence on $s$. 


\begin{theorem} \label{thm:bounded-2} There is a non-adaptive $k$-clustering algorithm which for any $\delta > 0$ and $2 \leq s \leq O(\sqrt{n})$, uses
\[
O\left(\frac{n^2}{s^2} k \log (n/\delta) \log \log n\right)
\]
subset queries of size at most $s$ and succeeds with probability $1-\delta$. \end{theorem}

For convenience, we will parameterize the query-size bound by $s = n^{1/r}$ where $r$ is any positive real number in the range $2 \leq r \leq \log n$. Before proving the theorem formally, we informally describe the algorithm and its analysis. A full description of the algorithm is given in pseudocode in \Alg{SB-bounded}, which is split into two phases: a "query selection phase", describing how queries are chosen by the algorithm, and a "reconstruction phase", describing how the algorithm uses the query responses to determine the clustering. Both phases contain a for-loop iterating over all $p \in \{0,1,\ldots,\log_r \log n - 1\}$ where the goal of the algorithm during the $p$'th iteration is to learn all remaining clusters of size at least $\smash{\frac{n}{k} \cdot 2^{-r^{p+1}}}$. We prove that this occurs with high probability in \Cref{lem:main-analysis}, which gives the main analysis. If each iteration is successful in doing so than the entire clustering has been learned successfully after iteration $p = \log_r \log n - 1$ (since $\smash{2^{-r^{\log_r \log n}} = 2^{-\log n} = \frac{1}{n}}$), and we justify this formally just after the statement of \Cref{lem:main-analysis-SB-s}. 

We describe the algorithm and it's analysis informally for the case of $r=2$, i.e. when the query sizes are bounded by $s = \sqrt{n}$. We also refer the reader to \Cref{sec:tech-2} for discussion of the ideas for the simple case of $k=3$. Consider some iteration $p \in \{0,1,\ldots,\log\log n - 1\}$ and suppose that prior to this iteration, all clusters of size at least $\frac{n}{k} \cdot 2^{-2^{p}}$ have been successfully recovered. Let $\cC_p$ denote the collection of all such clusters and let $\cR_p = \bigsqcup_{C \in \cC_p} C$ be the set of points they contain. The goal in iteration $p$ is to learn every cluster $C$ with $\smash{|C| \in [\frac{n}{k} \cdot 2^{-2^{p+1}}, \frac{n}{k} \cdot 2^{-2^{p}})}$. The algorithm queries $O(nk \log n)$ random sets $T$ formed by $2^{2^p}$ samples\footnote{Note that $p \leq \log\log n - 1$ and so $2^{2^p} \leq 2^{\frac{1}{2} \log n} = \sqrt{n}$.} from $U$ (see lines 5-7 of \Alg{SB-bounded}). If $T$ contains exactly two points $x,y \in T \setminus \cR_p$ belonging to unrecovered clusters, then we can use the fact that we already know the clustering on $\cR_p$ to tell whether or not $x,y$ belong to the same cluster or not, i.e. we can compute $\query(\{x,y\}) \in \{1,2\}$ from $\query(T)$. We then consider the set of all such pairs where $\query(\{x,y\}) = 1$ (this is $Q_p''$ defined in line 16) and consider the graph $G$ with this edge set, and vertex set $U \setminus \cR_p$, the set of points whose cluster hasn't yet been determined. If two points belong to the same connected component in this graph, then they belong to the same cluster. Thus, the analysis boils down to showing that with high probability, the induced subgraph $G[C]$ will be connected for every $C$ where $|C| \in [\frac{n}{k} \cdot 2^{-2^{p+1}}, \frac{n}{k} \cdot 2^{p})$. 

\paragraph{Proof of \Cref{thm:bounded-2}:} The following \Cref{lem:main-analysis-SB-s} establishes that after the first $p$ iterations of the algorithm's query selection and reconstruction phases, all clusters of size at least $\frac{n}{k} \cdot 2^{-r^{p+1}}$ have been learned with high probability. This is the main effort of the proof. After stating the lemma we show it easily implies that \Alg{SB-bounded} succeeds with probability at least $1-\delta$ by an appropriate union bound.

\begin{algorithm}
\caption{Algorithm with Optimal Dependence on Query Size} \label{alg:SB-bounded}
\textbf{Input:} Subset query access to a hidden partition $C_1 \sqcup \cdots \sqcup C_k = U$ of $|U| = n$ points\;
\emph{(Query Selection Phase)} \\
\For{$p = 0,1,\ldots,\log_r \log n - 1$} {
Initialize query set $Q_p \gets \emptyset$\;
\textbf{Repeat} $20 \cdot nk \ln (3 nk^2/\delta) \cdot 2^{r^{p+1}(1-\frac{2}{r})}$ times\;
$\longrightarrow$ Sample $T \subseteq U$ formed by $2^{r^p}$ independent uniform samples from $U$\; 
$\longrightarrow$ \textbf{Query} $T$ and add it to $Q_p$\; 
}
\emph{(Reconstruction Phase)} \\
Initialize learned cluster set $\cC_0 \gets \emptyset$\;
\For{$p = 0,1,\ldots,\log_r \log n - 1$} {
Let $\cC_p$ denote the collection of clusters reconstructed before iteration $p$\;
Let $\cR_p = \bigcup_{C \in \cC_p} C$ denote the points belonging to these clusters\;
Initialize $\cC_{p+1} \gets \cC_p$\;
Let $Q_p' = \{T \setminus \cR_p \colon T \in Q_p \text{ and } |T \setminus \cR_p| = 2\}$. Since each $T \in Q_p$ is a uniform random set, the elements of $Q_p'$ are uniform random pairs in $U \setminus \cR_p$\;
Let $Q_p'' = \{\{x,y\} \in Q_p' \colon \query(\{x,y\} = 1)\}$ denote the set of pairs in $Q_p'$ where both points lie in the same cluster. This set can be computed since $\query(T \setminus \cR_p) = \query(T) - \query(T \cap \cR_p)$ and $\query(T \cap \cR_p)$ is known since at this point we have reconstructed the clustering on $\cR_p$\;
Let $G_p$ denote the graph with vertex set $U \setminus \cR_p$ and edge set $Q_p''$\;
Let $\widetilde{C}_1,\ldots,\widetilde{C}_{\ell}$ denote the connected components of $G_p$ with size at least $\frac{n}{k} \cdot 2^{-r^{p+1}}$\;
Add $\widetilde{C}_1,\ldots,\widetilde{C}_{\ell}$ to $\cC_{p+1}$\;
}
\textbf{Output} clustering $\cC_{\log_r \log n}$\;
\end{algorithm}

\begin{lemma} \label{lem:main-analysis-SB-s} For each $p = 0,1,\ldots,\log_r \log n-1$, let $\cE_p$ denote the event that all clusters of size at least $\frac{n}{k} \cdot 2^{-r^{p+1}}$ have been successfully recovered immediately following iteration $p$ of \Alg{SB-bounded}. Then, 
\[
\Pr[\neg \cE_0] \leq \delta/k~ \text{ and } ~\Pr[\neg \cE_p ~|~ \cE_{p-1}] \leq \delta/k~ \text{ for all } p \in \{1,2,\ldots,\log_r \log n - 1\} \text{.}
\]
\end{lemma}

Before proving \Cref{lem:main-analysis-SB-s}, we observe that it immediately implies \Cref{thm:bounded-2} as follows. Let $I_0 = [\frac{n}{k} \cdot 2^{-r},n]$ and for $1 \leq p < \log_r\log n$, let $I_p = [\frac{n}{k} \cdot 2^{-r^{p+1}},\frac{n}{k} \cdot 2^{-r^p})$. If there are no clusters $C$ for which $|C| \in I_p$, then trivially $\Pr[\neg \cE_p ~|~ \cE_{p-1}] = 0$, and otherwise $\Pr[\neg \cE_p ~|~ \cE_{p-1}] \leq \delta/k$ by the lemma. Since there are $k$ clusters, clearly there are at most $k$ values of $p$ for which there exists a cluster with size in the interval $I_p$. Using this observation and a union bound, we have
\[
\Pr[\neg \cE_{\log_r \log n - 1}] \leq \Pr[\neg \cE_0] + \sum_{p=1}^{\log_r \log n} \Pr[\neg \cE_p ~|~ \cE_{p-1}] \leq \delta
\]
which completes the proof of correctness since the algorithm succeeds iff $\cE_{\log_r \log n - 1}$ occurs. 

\paragraph{Query complexity:} The total number of queries made by \Alg{SB-bounded} is $O(nk \log (n/\delta)) \cdot \sum_{p=1}^{\log_r\log n} 2^{r^{p}(1-\frac{2}{r})}$ and 
\[
\sum_{p=1}^{\log_r\log n} 2^{r^{p}(1-\frac{2}{r})} \leq \log_r \log n \cdot 2^{r^{\log_r \log n}(1-\frac{2}{r})} = \log_r \log n \cdot n^{1-\frac{2}{r}} \leq \log \log n \cdot \frac{n}{s^2}
\]
which completes the proof. \qed \smallskip 

\begin{proof} \emph{of \Cref{lem:main-analysis-SB-s}.} Let $\cR_p$ denote the set of points belonging to a cluster which has been recovered before iteration $p$.  

\noindent \textbf{Case 1:} $p=0$. In this iteration, the algorithm queries $|Q_0| \geq 8 \cdot nk \ln (3 nk^2/\delta) \cdot 2^{r-2}$ random pairs and we need to show that it successfully recovers all clusters with size at least $\frac{n}{k \cdot 2^r}$ with probability at least $1-\delta/k$. Let $C$ denote any such cluster and recall from lines (16-17) the definition of the graph $G_0$ with vertex set $U$ and edge set $Q_0''$. We will show that the induced subgraph $G_0[C]$ is connected, and thus $C$ is correctly recovered in lines (18-19), with probability at least $1-\delta/k^2$. Since there are at most $k$ clusters, the lemma holds by a union bound. 

Consider any two vertices $x,y \in C$ and note that $|Q_0| \geq \frac{2n^2 \ln (3 nk^2/\delta)}{|C|}$ since $|C| \geq \frac{n}{k \cdot 2^r}$. We lower bound the probability that $(x,y)$ is an edge in $G_0[C]$ as follows. Note that this occurs iff $\{x,y\} \in Q_0$. Thus,
\begin{align} \label{eq:p=0_bound}
    \Pr_{Q_0}[(x,y) \in E(G_0[C])] &= \Pr_{Q_0}[\{x,y\} \in Q_0] = 1-\left(1-\frac{1}{n^2}\right)^{|Q_0|} \nonumber \\
    &\geq 1-\exp\left(-\frac{2\ln (3 nk^2/\delta)}{|C|}\right) \nonumber \\
    &\geq 1-\exp\left(-\frac{2\ln (3 k^2 |C|/\delta)}{|C|}\right) = 1-\left(\frac{\delta}{3 k^2|C|}\right)^{\frac{2}{|C|}}
\end{align}
and so by \Cref{fact:random-connected}, $G_0[C]$ is connected with probability at least $1-\delta/k^2$ as claimed.

\noindent \textbf{Case 2:} $1 \leq p < \log_r \log n$. Recall from lines (12-13) that $\cC_p$ denotes the set of clusters recovered prior to iteration $p$ and $\cR_p = \bigcup_{C \in \cC_p} C$ is the set of points belonging to these clusters. Note that we are conditioning on the event that every cluster of size at least $\frac{n}{k} \cdot 2^{-r^p}$ has been recovered prior to iteration $p$. Let $C$ denote some cluster with size 
\[
|C| \in \left[\frac{n}{k} \cdot 2^{-r^{p+1}},\frac{n}{k}\cdot 2^{-r^{p}}\right) \text{ and note that } |U \setminus \cR_p| \leq k \cdot \frac{n}{k} \cdot 2^{-r^{p}} = n \cdot 2^{-r^{p}} \text{.}
\]
Recall from lines (16-17) the definition of $Q_p''$ and that $G_p$ is the graph with vertex set $U \setminus \cR_p$ and edge set $Q_p''$. We need to argue that the induced subgraph $G_p[C]$ is connected, and thus $C$ is correctly recovered in lines (18-19), with probability at least $1-\delta/k^2$. Since there are at most $k$ clusters, a union bound completes the proof of the lemma. 

Consider any two vertices $x,y \in C$. We lower bound the probability that $(x,y)$ is an edge in $G_p[C]$, which occurs iff there is some $T \in Q_p$ where $T \setminus \cR_p = \{x,y\}$. We have
\begin{align}
    \Pr_{T \colon |T| = 2^{r^p}}[T \setminus \cR_p = \{x,y\}] = {2^{r^p} \choose 2} \cdot \frac{1}{n^2} \cdot \left(\frac{|\cR_p|}{n}\right)^{t-2} \geq \frac{2^{2r^p}}{3n^2}  \left(1-2^{-r^p}\right)^{t} \geq \frac{2^{2r^p}}{10 n^2} \nonumber 
\end{align} 
and since $|Q_p| = 20 nk \ln (3 nk^2/\delta) \cdot 2^{r^{p+1}(1-\frac{2}{r})}$, we have
\begin{align}
    \Pr_{Q_p}\left[(x,y) \in E(G_p[C])\right] &= \Pr_{Q_p}\left[\{x,y\} \in Q_p''\right] = \Pr_{Q_p}\left[\exists T \in Q_p \colon T \setminus \cR_p = \{x,y\}\right] \nonumber \\
    &\geq 1 - \left(1-\frac{2^{2r^p}}{10 n^2}\right)^{20 \cdot nk \cdot 2^{r^{p+1}-2r^p}\ln (3nk^2/\delta)} \nonumber \\
    &\geq 1-\exp\left(-\frac{2\cdot 2^{r^{p+1}}k \ln (3nk^2/\delta)}{n}\right) \nonumber
\end{align}
and plugging in $|C| \geq \frac{n}{k} \cdot 2^{-r^{p+1}}$ and $|C| \leq n$ into the RHS yields
\begin{align} 
    \Pr_{Q_p}\left[(x,y) \in E(G_p[C])\right] &\geq 1-\exp\left(-\frac{2\ln (3nk^2/\delta)}{|C|}\right) \nonumber \\
    &\geq 1-\exp\left(-\frac{2\ln (3k^2|C|/\delta)}{|C|}\right) = 1-\left(\frac{\delta}{3k^2|C|}\right)^{\frac{2}{|C|}}\text{.} \nonumber
\end{align}
Therefore, $(x,y)$ is an edge in $G_p[C]$ with probability at least $1-\left(\frac{\delta}{3k^2|C|}\right)^{\frac{2}{|C|}}$, which by \Cref{fact:random-connected} implies that $G_p[C]$ is connected with probability at least $1-\delta/k^2$ as claimed. This completes the proof of \Cref{thm:bounded-2}. \end{proof}

\subsection{An $\Omega(n^2/s^2)$ Lower Bound for $k \geq 3$} \label{sec:LB}

\begin{theorem} \label{thm:3-2-NA-LB} Non-adaptive $3$-clustering requires $\Omega(n^2)$ pair queries. \end{theorem}

\begin{proof} For every $(x,y) \in {U \choose 2}$ consider the following pair of partitions:
\[
P^1_{x,y} = (\{x,y\},\emptyset,U \setminus \{x,y\}) \text{ and } P^2_{x,y} = (\{x\},\{y\},U \setminus \{x,y\}) \text{.}
\]
Observe that $\query$ returns the same value for $P^1_{x,y}$ and $P^2_{x,y}$ on every possible pair query except $\{x,y\}$. Thus, if query set $Q \subseteq U \times U$ distinguishes these two clusterings, then $Q \ni \{x,y\}$. Therefore, the number of pairs $\{x,y\}$ such that $Q$ distinguishes $P^1_{x,y}$ and $P^2_{x,y}$ is at most $|Q|$. Now, let $A$ be any non-adaptive pair-query algorithm which successfully recovers an arbitrary $3$-clustering with probability $\geq 2/3$. The algorithm $A$ queries a random set $Q \subseteq U \times U$ according to some distribution, $\cD_A$. In particular, for every $\{x,y\} \in {U \choose 2}$, $Q$ distinguishes $P^1_{x,y}$ and $P^2_{x,y}$ with probability $\geq 2/3$. Thus,
\begin{align}
    \frac{2}{3} {n \choose 2} &\leq \sum_{\{x,y\} \in {U \choose 2}} \Pr_{Q \sim \cD_A}[Q \text{ distinguishes } P^1_{x,y} \text{ and } P^2_{x,y}] \nonumber \\
    &= \Exp_{Q\sim \cD_A}\left[ \left|\left\{\{x,y\} \in {U \choose 2} \colon Q \text{ distinguishes } P^1_{x,y} \text{ and } P^2_{x,y} \right\}\right| \right] \leq |Q| \nonumber
\end{align}
using linearity of expectation, and this completes the proof. \end{proof}

\begin{corollary} \label{cor:3-s-LB} Non-adaptive $3$-clustering requires $\Omega(n^2/s^2)$ subset queries of size at most $s$. \end{corollary}

\begin{proof} This follows from \Cref{thm:3-2-NA-LB} since one $s$-sized query can be simulated by ${s \choose 2}$ pair-queries. \end{proof}

Thus, in order to achieve a near-linear non-adaptive upper bound for $3$-clustering, we require an algorithm which makes queries of size $\widetilde{\Omega}(\sqrt{n})$.

\section{The Special Case of Balanced Clusters} \label{sec:balanced}

In this section we present improved non-adaptive algorithms for recovering a $k$-clustering with subset queries when the clusters are roughly balanced according to the following notion.

\begin{definition} Given $|U| = n$ and $B \geq 1$, we say that a $k$-partition $C_1 \sqcup \cdots \sqcup C_k = U$ is $B$-balanced if $\frac{n}{B k} \leq |C_i| \leq \frac{B n}{k}$ for all $i \in [k]$. \end{definition}

Given constant $B > 1$ and failure probability $\delta = \poly^{-1}(k)$, our first algorithm in \Cref{sec:bal-1} makes $O(n \log k + k \log^4 k)$ queries and our second algorithm in \Cref{sec:bal-2} makes $O(n \log^2 k)$ queries. In particular, as long as $\smash{k = O(\frac{n}{\log^3 n})}$, the algorithm of \Cref{sec:bal-1} makes $O(n \log k)$ queries. Note that we also describe a two-round algorithm for this setting making $O(n\log\log k)$ queries in \Cref{sec:2round-2}. 

\subsection{An $O(n \log k) + \widetilde{O}(k)$ Algorithm} \label{sec:bal-1}

\begin{theorem} \label{thm:k-bal-1} There is a non-adaptive algorithm that recovers a $B$-balanced $k$-clustering using $O(B^2 n \log (k/\delta)) + O(B k \log^4 (k/\delta))$ subset queries of size $O(k \log k)$ and succeeds with probability at least $1-\delta$. In particular, for $\delta = 1/\poly(k)$, any constant $B \geq 1$, and $k = O(n/\log^3 n)$, the query complexity is $O(n \log k)$. 
\end{theorem}

Pseudocode for the algorithm is given in \Alg{bal-1}. In line (3) we draw $q = \Theta(B^2 \log k)$ sets $T_1,\ldots,T_q$ each formed by $k/B$ samples from $U$ and in line (5) learn the clustering over their union using \Cref{thm:1}. I.e., for $T = T_1 \cup \cdots \cup T_q$, we find $R_j = T \cap C_j$. Then, we query $T_i$ and $T_i \cup \{x\}$ for every $x \in U$ and every $i \in [q]$ in line (5). Now, consider some point $x \in U$ and let $j^{\ast}$ be it's cluster's index. Note that $\query(T_i \cup \{x\}) = \query(T_i)$ iff $T_i$ intersects $C_{j^{\ast}}$. Thus, if $T_i$ does not intersect $C_{j^{\ast}}$, then every cluster $j$ that $T_i$ intersects can be ruled out as a candidate for being the cluster containing $x$. The set $J_x$ computed in line (8) is the set of all $j$ which can be ruled out in this way. If for every $j \neq j^{\ast}$, there is some $T_i$ containing $j$, but not $j^{\ast}$, then $J_x = \{j^{\ast}\}$ and we determine $j^{\ast}$ in line (9). This occurs \emph{for every $x \in U$} if the following holds: for every pair $(j,j') \in {U \choose 2}$, there exists $T_i$ intersecting $C_j$, but not $C_{j'}$. We show in \Cref{clm:miss-bound-1} that this happens with high probability.

\paragraph{Proof of \Cref{thm:k-bal-1}:} There are $O(B^2 n \log (k/\delta))$ queries made in line (5) and $O(B \cdot k \log^4 (k/\delta))$ queries in line (4), since $|\bigcup_{i\in [q]} T_i| = O(B k \log (k/\delta))$. We now prove correctness, which is due to the following claim.



\begin{algorithm}
\caption{First Algorithm for Balanced Clustering} \label{alg:bal-1}
\textbf{Input:} Subset query access to a $B$-balanced partition $C_1 \sqcup \cdots \sqcup C_k = U$ of $|U| = n$ points\;
\emph{(Query Selection Phase)} \\
Choose $q = 2eB^2 \ln (2k^2/\delta)$ sets $T_{1},\ldots,T_{q}$ each formed by $\frac{k}{B}$ uniform samples from $U$\;
Run the algorithm from \Cref{thm:1} with error probability $\delta/2$ to learn the clustering restricted on $\widetilde{C} = \bigcup_{i=1}^q T_i$. Let $(\widetilde{C}_1,\ldots,\widetilde{C}_k)$ be the clustering returned by the algorithm. I.e., if the algorithm is successful, then $\widetilde{C}_j = \widetilde{C} \cap C_j$ for all $j \in [k]$\; 
\textbf{Query} $T_{i}$ and $T_{i} \cup \{x\}$ for all $i \in [q]$ and all $x \in U$\;
\emph{(Reconstruction Phase)} \\
\For{$x \in U$} {
Let $J_x = \bigcup_{i \in [q] \colon \query(T_{i} \cup \{x\}) \neq \query(T_{i})} \{j \in [k] \colon T_{i} \cap \widetilde{C}_j \neq \emptyset\}$. Note that $T_i \cap \widetilde{C}_j \neq \emptyset$ iff $T_i \cap C_j \neq \emptyset$. Note that $\query(T_{i} \cup \{x\}) \neq \query(T_{i})$ iff $x$ does not belong to any cluster that is hit by $T_{i}$. Thus, $J_x$ is the collection of all $j$ such that some set $T_{i}$ has revealed that $x \notin C_j$\;
\eIf{$|J_x| = k-1$} {
Add $x$ to $\widetilde{C}_{j^{\ast}}$ where $j^{\ast}$ is the unique element of $[k] \setminus J_x$\;
}
{
\textbf{Output} fail\;
}
}
\textbf{Output} clustering $(\widetilde{C}_1,\ldots,\widetilde{C}_k)$\;
\end{algorithm}

\begin{claim} \label{clm:miss-bound-1} For $i \in [q], j \in [k]$, let $\cE_{i,j}$ denote the event that $T_i \cap C_j \neq \emptyset$. Then, 
\begin{align} \label{eq:miss-bound}
    \Pr_{T_1,\ldots,T_q}\left[\forall (j,j') \in {[k] \choose 2} \text{, } \exists i \in [q] \colon \cE_{i,j} \wedge \neg \cE_{i,j'}\right] \geq 1-\delta/2 \text{.}
\end{align}    
\end{claim}
\begin{proof} Firstly, for fixed $i \in [q]$ and $j \neq j'$, since each cluster's size is bounded in the interval $[\frac{n}{Bk},\frac{Bn}{k}]$, we have
\begin{align}
    \Pr_{T_i}[\cE_{i,j} \wedge \neg \cE_{i,j'}] &= \Pr[\cE_{i,j}] \cdot \Pr[\neg \cE_{i,j'} ~\lvert~ \cE_{i,j}] \nonumber \\
    &= \left(1 - \left(1-\frac{|C_j|}{n}\right)^{|T_i|} \right) \cdot \left(1-\frac{|C_{j'}|}{n}\right)^{|T_i|-1} \nonumber \\ 
    &\geq  \left(1 - \left(1-\frac{1}{Bk}\right)^{k/B} \right) \cdot \left(1-\frac{B}{k}\right)^{k/B} \geq \left(1 - \exp\left(B^{-2}\right)\right) \cdot \frac{1}{e} \geq \frac{1}{2eB^2} \nonumber
\end{align}
and so for a fixed $(j,j') \in {[k] \choose 2}$, we have
\[
\Pr_{T_1,\ldots,T_q}\left[\forall i \in [q] \colon \neg \left(\cE_{i,j} \wedge \neg \cE_{i,j'}\right) \right] \leq \left(1-\frac{1}{2eB^2}\right)^{2eB^2 \ln (2k^2/\delta)} \leq \frac{\delta}{2k^2}
\]
and the claim follows by a union bound over all $(j,j') \in {[k] \choose 2}$. \end{proof}
    
By \Cref{clm:miss-bound-1}, with probability at least $1-\delta/2$, for every $j \neq j' \in [k]$ we have some $T_i$ such that $T_i \cap C_j \neq \emptyset$ and $T_i \cap C_{j'} = \emptyset$. In particular, for $x \in U$, let $C_{j^{\ast}}$ be the cluster containing $x$. For every $j \neq j^{\ast}$ we have some $T_i$ such that $T_i \cap C_j \neq \emptyset$ and $T_i \cap C_{j^{\ast}} = \emptyset$ which means that in line (9) of the algorithm, we have $J_x = [k] \setminus \{j^{\ast}\}$ and so we successfully identify the cluster containing $x$. Moreover, this occurs for all $x$. Finally, line (4) succeeds with probability $\delta$ and thus the entire algorithm succeeds with probability at least $1-\delta$ by a union bound. 

\subsection{An $O(n \log^2 k)$ Algorithm} \label{sec:bal-2}

We now give an alternative algorithm which improves on \Cref{thm:k-bal-1} when $k \gg n/\log^3 n$.

\begin{theorem} \label{thm:k-bal-2} There is a non-adaptive algorithm for recovering a $B$-balanced $k$-clustering using $O(B^2 n \log (k/\delta) \log k)$ subset queries of size $O(k)$ which succeeds with probability $1-\delta$. In particular, for failure probability $\delta = 1/\poly(k)$ and any constant $B \geq 1$, the query complexity is $O(n \log^2 k)$. \end{theorem}

\begin{proof} We will use the following subroutine which we introduced in \Cref{sec:tech-3}.

\corollaryonebit*

The pseudocode for the algorithm is given in \Alg{bal-2}. In words, \Cref{cor:1-bit-1} says that if we have a set $R$ containing exactly one representative from $C(x)$, then with $O(\log |R|)$ subset queries we can identify that representative. Thus, suppose we have a collection of sets $R_1,\ldots,R_q$ such that for every cluster $j \in [k]$, there is some $R_i$ containing a unique representative from $C_j$. Consider the bipartite graph where on the left we have $U$ and on the right we have $R_1 \cup \cdots \cup R_q$. Then, for every $x \in U$ and every $R_i$ we can run the procedure from \Cref{cor:1-bit-1}, and if it returns a representative $y \in R_i \cap C(x)$, then we add the edge $(x,y)$ to this graph. By the property of $R_1,\ldots,R_q$, two vertices $x,y \in U$ belong to the same cluster iff they are connected by a path of length $2$ in this graph. We show that setting $q = \Theta(B^2 \log (k/\delta))$ and letting each $R_i$ be a random sample of $k/B$ elements from $U$ results in a collection of sets with this good property with high probability. This leads to a query complexity of $n \cdot q \cdot O(\log k) = O(n \log(k/\delta) \log k)$. 

\begin{algorithm}
\caption{Second Algorithm for Balanced Clustering} \label{alg:bal-2}
\textbf{Input:} Subset query access to a $B$-balanced partition $C_1 \sqcup \cdots \sqcup C_k = U$ of $|U| = n$ points\;
Choose $q = eB^2 \ln (k/\delta)$ sets $R_{1},\ldots,R_{q}$ each formed by $\frac{k}{B}$ uniform samples from $U$\;
Construct a bipartite graph $G(U,\bigcup_{j=1}^{q} R_j,E)$ as follows\; \For{$x \in U$ and $i \in [q]$} {
Run the algorithm $\mathsf{FINDREP}$ from \Cref{cor:1-bit-1} on $(x,R_i)$\;
\If{the algorithm certifies there is a unique $y \in R_i$ such that $x,y$ are in the same cluster} {
Add the edge $(x,y)$ to $E(G)$\;
}
}
Let $\widetilde{C}_1,\ldots,\widetilde{C}_{\ell}$ denote the connected components of $G$\;
\textbf{Output} the clustering $(\widetilde{C}_1,\ldots,\widetilde{C}_{\ell})$\;
\end{algorithm}



The algorithm makes $n \cdot q \cdot O(\log (k/B)) = O(B^2 n \log(k/\delta) \log k)$ queries. The correctness follows immediately from the following claim. \end{proof}

\begin{claim} \label{clm:reps} With probability at least $1-\delta$, for every $j \in [k]$, there exists $i \in [q]$ such that $|R_i \cap C_j| = 1$. \end{claim}

\begin{proof} Fix $j \in [k]$ and $i \in [q]$. We have
\[
\Pr[|R_i \cap C_j| = 1] = |R_i| \cdot \frac{|C_j|}{n} \cdot \left(1-\frac{|C_j|}{n}\right)^{|R_i|-1} \geq \frac{k}{B} \cdot \frac{1}{B k} \cdot \left(1-\frac{B}{k}\right)^{k/B} \geq \frac{1}{e B^2}
\]
and so for a fixed $j \in [k]$,
\[
\Pr[\forall i \in [q] \colon |R_i \cap C_j| \neq 1] \leq \left(1-\frac{1}{eB^2}\right)^{e B^2 \ln (k/\delta)} \leq \delta/k
\]
and so by a union bound
\[
\Pr[\exists j \in [k] \text{, } \forall i \in [q] \colon |R_i \cap C_j| \neq 1] \leq \delta
\]
and this completes the proof. \end{proof}

\section{Algorithms with Two Rounds of Adaptivity} \label{sec:2round}

In this section we describe two algorithms that use \emph{two rounds} of adaptivity. That is, these algorithms are allowed to specify a round of queries, receive the responses, perform some computation, then specify a second round of queries and receive the responses, before finally recovering the clustering. We give a simple \emph{deterministic} algorithm making $O(n \log k)$ queries in \Cref{sec:2round-1} and a randomized algorithm for recovering a balanced clustering with $O(n \log \log k)$ queries in \Cref{sec:2round-2}. Both algorithms exploit the additional round of queries to first compute a set containing exactly one representative from every cluster. 

\subsection{A $2$-Round $O(n \log k)$ Deterministic Algorithm} \label{sec:2round-1}

\begin{theorem} \label{thm:2-round} There is a $2$-round, non-adaptive, deterministic $k$-clustering algorithm using $O(n \log k)$ subset queries. 
\end{theorem}

\begin{proof} Pseudocode for the algorithm is given in \Alg{2round}. In the first round (lines 3-5), we use $n$ queries to compute a set $R \subseteq U = \{x_1,\ldots,x_n\}$ containing exactly one representative from every cluster. This is done by querying every prefix $P_t = \{x_1,\ldots,x_t\}$ and observing that $\query(P_t) - \query(P_{t-1}) = 1$ iff $x_t$ is the only representative for its cluster in $P_t$. Thus, the set $R$ computed in line (4) contains, for each cluster $C$, the first member of $C$ in the ordering $x_1,\ldots,x_n$. 
The second round of queries is used to determine, for every $x \in U$, the unique representative of $C(x)$ in $R$ (see line 8). To accomplish this we use the following subroutine which we introduced in \Cref{sec:tech-3}.

\corollaryonebit*

\begin{algorithm}
\caption{Deterministic $2$-Round Algorithm} \label{alg:2round}
\textbf{Input:} Subset query access to a hidden partition $C_1,\ldots,C_k$ of $U = \{x_1,\ldots,x_n\}$\;
\emph{Round 1:} \\
\textbf{Query} $P_t = \{x_i \colon i \leq t\}$ for every $t \in [n]$\;
Define $R = \{x_t \colon \query(P_t) - \query(P_{t-1}) = 1\}$ containing exactly one point from every cluster\;
For each $y \in R$, define cluster $\widetilde{C}_y = \{y\}$\;
\emph{Round 2:} \\ 
\For{$x \in U$} {
Run the algorithm $\mathsf{FINDREP}$ from \Cref{cor:1-bit-1} on $(x,R)$, which is guaranteed to output the unique $y \in R$ for which $x,y$ lie in the same cluster\;
Place $x$ into $\widetilde{C}_y$\;
}
\textbf{Output} clustering $(\widetilde{C}_y \colon y \in R)$\;
\end{algorithm}

By \Cref{cor:1-bit-1}, the number of queries made in the second round is clearly $O(n \log k)$. \end{proof}


\subsection{A $2$-Round $O(n\log\log k)$ Algorithm for Balanced Clusters} \label{sec:2round-2}

Recall that a clustering $C_1 \sqcup \cdots \sqcup C_k = U$ is $B$-balanced if $\frac{n}{Bk} \leq |C_j| \leq \frac{Bn}{k}$. 

\begin{theorem} \label{thm:2-round-bal} There is a $2$-round, non-adaptive algorithm which recovers a $B$-balanced $k$-clustering using 
\[
O\left(\sqrt{B} \cdot n \log \Big(\frac{\log k}{\delta}\Big)\right)
\]
subset queries and succeeds with probability $1-\delta$. In particular, if $B$ is a constant and $\delta = 1/\poly\log k$, then the query complexity is $O(n \log \log k)$. \end{theorem}

\begin{proof} We will use the following result of \cite{GK00} on query-based reconstruction of bipartite graphs as a black-box. Given a bipartite graph $G(V,W,E)$, an edge-count query on $(S,T)$ where $S \subseteq V$, $T \subseteq W$ returns $|E \cap S \times T|$, the number of edges between $S$ and $T$.
     
\begin{lemma} [\cite{GK00}, see Section 4.3] \label{lem:GK-matching} There is a deterministic, non-adaptive algorithm which reconstructs any bipartite graph $G(V,W,E)$ where (a) $|V| = n$, (b) $|W| = m$, and (c) every vertex in $V$ has degree at most $1$, using $O(n \cdot \frac{\log m}{\log n})$ edge-count queries. \end{lemma} 

We will say a set $A \subseteq U$ is an \emph{independent set} if each element of $A$ belongs to a distinct cluster. Given two independent sets $A,B$ let $M(A,B)$ be the matching where there is an edge from $x \in A$ to $y \in B$ if $x,y$ belong to the same cluster. We observe that edge-count queries in $M(A,B)$ can be simulated by subset queries, leading to the following corollary.

\begin{corollary} \label{cor:matching} Suppose that $A,B \subseteq U$ are independent sets. There is a deterministic, non-adaptive algorithm $\mathsf{IS}$-$\mathsf{MATCH}$ which reconstructs $M(A,B)$ using $O(|A| \cdot \frac{\log |B|}{\log |A|})$ subset queries. \end{corollary}

\begin{proof} We need to show that an edge-count query $(S,T)$ where $S \subseteq A$, $B \subseteq T$ can be simulated by a constant number of subset queries. Let $m(S,T)$ denote the number of edges in $M(A,B)$ between $S$ and $T$. Since $A,B$ are independent sets, $S,T$ are also independent sets, and so we have 
\[
m(S,T) = \query(S) + \query(T) - \query(S \cup T)
\]
since $m(S,T)$ is the number of clusters intersected by both $S$ and $T$. Thus, one edge-count query to $M(A,B)$ can be simulated by three subset queries and this completes the proof. \end{proof}

Pseudocode for the algorithm is given in \Alg{2round-bal}. The algorithm is parameterized in terms of a value $\tau > 1$ which we will choose later in the proof so as to minimize the query complexity. The first round is used to accomplish the following. In lines (4-5) we construct a set $R$ containing exactly one representative from every cluster and use this to define an initial clustering. In line (6) we sample random sets $I_1,\ldots,I_s$ and in line (8) make a query to each to check whether or not it is an independent set. Line (10) defines $V$ which is the union of all the $I_i$'s which are independent sets. We now describe the second round. In line (14) we run the procedure of \Cref{cor:matching} to construct the matching $M(I_i,R)$ whenever $I_i$ is an independent set. Finally, we determine for every $x \in U$, the unique $y \in R$ for which $x,y$ belong to the same cluster. If $x \in V$ this is done in lines (18-20) by taking $x$'s neighbor in $M(I_i,R)$ for some independent set $I_i$. If $x \notin V$, this is done in lines (23-24) by running the $\mathsf{FINDREP}$ procedure of \Cref{cor:1-bit-1}. 

\begin{algorithm}
\caption{Randomized $2$-Round Algorithm for Balanced Clustering} \label{alg:2round-bal}
\textbf{Input:} Subset query access to a hidden partition $C_1 \sqcup \cdots \sqcup C_k = U$ of $|U| = n$ points\;
\emph{Round 1:} \\
\textbf{Query} $P_t = \{x_i \colon i \leq t\}$ for every $t \in [n]$\;
Define $R = \{x_t \colon \query(P_t) - \query(P_{t-1}) = 1\}$ containing exactly one element of every cluster\;
For each $y \in R$, define initial cluster $\widetilde{C}_y = \{y\}$\;
Sample $\smash{s = 10 \sqrt{\frac{B}{k}} \cdot n \ln (\tau/\delta)}$ sets $I_1,\ldots,I_s \subset U$ each formed by $\sqrt{\frac{k}{10B}}$ independent uniform samples from $U$\;
\For{$i \in [s]$}
{
\textbf{Query} $I_i$. (This is to check if $\query(I_i) = |I_i|$, i.e. whether $I_i$ is an independent set.)\;
}
Let $V = \bigcup_{i \in [s] \colon \query(I_i) = |I_i|} I_i$ be the points in $U$ lying in an independent set among $I_1,\ldots,I_s$\;
If $|V| < n(1-\frac{1}{\tau})$, then \textbf{output fail}. Otherwise, continue\;
\emph{Round 2:} \\ 
\For{$i \in s \colon \query(I_i) = |I_i|$}
{
Run the algorithm $\mathsf{IS}$-$\mathsf{MATCH}$ of \Cref{cor:matching} on sets $(I_i,R)$ and let $M_i \subset I_i \times R$ be the output\;
(Note: by construction of $R$, \Cref{cor:matching} guarantees that every element of $I_i$ is matched by $M_i$)\;
}
\For{$x \in U$}
{
\If{$x \in V$}
{
Choose $I_i$ such that $x \in I_i$ and $I_i$ is an independent set\;
Let $y \in R$ denote the neighbor of $x$ in the matching $M_i \subset I_i \times R$\;
Place $x$ into $\widetilde{C}_y$\;
}
\If{$x \in U \setminus V$} {
Run the algorithm $\mathsf{FINDREP}$ of \Cref{cor:1-bit-1} to find the unique $y \in R$ for which $x,y$ lie in the same cluster\;
Place $x$ into $\widetilde{C}_y$\;
}
}
\textbf{Output} clustering $(\widetilde{C}_y \colon y \in R)$\;
\end{algorithm}

The algorithm always either outputs fail in line (11), or correctly reconstructs the clustering by \Cref{cor:matching} and \Cref{cor:1-bit-1}. Thus we only need to argue that $|U \setminus V| \leq \frac{n}{\tau}$ occurs with probability at least $1-\delta$ allowing it to pass the check in line (11), and that conditioned on this, the algorithm makes $O(n \ln \ln k)$ queries when we set $\tau$ appropriately. Let us first count the number of queries conditioned on this event. Line (8) performs $s$ queries. Since each $I_i$ is of size $\sqrt{k}$ and $|R| = k$, by \Cref{cor:matching}, lines (13-14) perform a total of $O(s \cdot \sqrt{k}) = O(\sqrt{B} \cdot n \ln (\tau/\delta))$ queries. Lines (23-24) use $|U \setminus V| O(\log k) = O(\frac{n}{\tau} \log k)$ queries. Setting $\tau = \ln k$ yields a query complexity of $O(\sqrt{B} n \ln (\frac{\ln k}{\delta}))$. We now prove in \Cref{clm:Vlarge} that the required bound on $|U \setminus V|$ holds with high probability, and this completes the proof. \end{proof}

\begin{claim} \label{clm:Vlarge} With probability at least $1-\delta$, we have $|U \setminus V| \leq \frac{n}{\tau}$. \end{claim}

\begin{proof} We prove an appropriate bound on $\Exp[|U \setminus V|]$ and then apply Markov's inequality. Fix $x\in U$. For $i \in [s]$, let $\cE_{x,i}$ denote the event that $x \in I_i$ and $I_i$ is an independent set. Observe that $x \in U \setminus V$ iff $\cE_{x,i}$ does not occur for every $i \in [s]$. We first lower bound the probability of $\cE_{x,i}$. Observe that
\begin{align} \label{eq:Exi}
    \Pr_{I_i}[\cE_{x,i}] = \Pr[x \in I_i] \Pr[I_i \text{ an independent set } ~|~ x \in I_i]
\end{align}
and 
\[
\Pr_{I_i}[x \in I_i] = 1-\left(1-\frac{1}{n}\right)^{|I_i|} \geq 1-\exp\left(-\frac{|I_i|}{n}\right) \geq \frac{|I_i|}{2n} \geq \sqrt{\frac{k}{B}} \cdot \frac{1}{8n}
\]
where we have used the inequality $\exp(-z) \leq 1-\frac{z}{2}$ for $z \in [0,1]$. Next, by a simple union bound over all pairs in $I_i$ and the fact that every cluster is bounded as $|C_j| \leq \frac{Bn}{k}$, we have
\[
\Pr[I_i \text{ not an independent set } ~|~ x \in I_i] \leq |I_i|^2 \frac{B}{k} \leq \frac{1}{10} \text{.}
\]
Plugging these bounds back into \Cref{eq:Exi} yields $\Pr_{I_i}[\cE_{x,i}] \geq \sqrt{\frac{k}{B}} \cdot \frac{1}{10n}$ and noting that these events are independent due to the $I_i$'s being independent yields
\begin{align}
    \Pr[x \notin V] = \Pr[\neg \cE_{x,i},~ \forall i \in [s] ] \leq \left(1-\sqrt{\frac{k}{B}} \cdot \frac{1}{10n}\right)^s = \exp(-\ln (\tau/\delta)) = \delta/\tau
\end{align}
where we have used the definition of $s = 10 \sqrt{B/k} \cdot n \ln (\tau/\delta)$. Finally, this implies $\Exp[|U\setminus V|] \leq \delta n/\tau$ and so by Markov's inequality $\Pr[|U\setminus V| > \frac{n}{\tau}] < \delta$. This completes the proof. \end{proof}

\bibliographystyle{alpha}
\bibliography{biblio}

\appendix

\section{Useful Lemmas}

\subsection{Proofs of Group Testing Primitives} \label{sec:supp-recovery}

Given $v \in \{0,1\}^n$, let $\supp(v) = \{i \colon x_i = 1\}$ denote the support of $v$. An $\mathsf{OR}$-query on set $S \subseteq [n]$ returns 
\[
\mathsf{OR}_S(v) = \bigvee_{i \in S} v_i = \mathbf{1}\left(\supp(x) \cap S \neq \emptyset\right) \text{.}
\]
This section discusses the problem of recovering the support of a vector via $\mathsf{OR}$ queries. In particular, we are interested in \emph{non-adaptive} algorithms for this problem. The results in this section are standard in the combinatorial group testing and coin-weighing literature. See e.g. \cite{HS87,PR08} and also \cite{assadi2021graph}, who applied these results to obtain query algorithms for graph connectivity.


\begin{lemma} \label{lem:1-bit} Let $v \in \{0,1\}^n$ such that $|\supp(v)| = 1$. There is a deterministic, non-adaptive algorithm that makes $\lceil \log n \rceil$ $\mathsf{OR}$ queries and returns $\supp(v)$. \end{lemma}

\begin{proof} Since $|\supp(v)| = 1$, an $\mathsf{OR}$ query on set $S$ is equivalent to taking $\langle v, w \rangle$ where $w_i = 1$ iff $i \in S$. Let $M$ be the $\lceil \log n \rceil \times n$ matrix whose $i$'th column is simply $b^i \in \{0,1\}^{\lceil \log n \rceil}$, the binary representation of $i$. The rows of $M$ correspond to $\mathsf{OR}$ queries. Then, $Mv = \sum_{i = 1}^n x_i b^i = \sum_{i \colon x_i = 1} b^i = b_j$ where $j$ is the unique coordinate where $x_j = 1$. \end{proof}

\onebitcertone*

\begin{proof} Let $M$ be the $\lceil \log n \rceil \times n$ matrix described in the proof of \Cref{lem:1-bit}. Let $\mathbf{1} = 1^{\lceil \log n \rceil \times n}$ denote the all $1$'s matrix with the same dimensions. We query $M \cdot v$ and $(\mathbf{1} - M) \cdot v$ where here $(\cdot)$ denotes the "$\mathsf{OR}$ product". I.e. the $i$'th coordinate of $M \cdot v$ is $\mathbf{1}((Mv)_i > 0)$. Note that $\mathbf{1} - M$ is obtained by flipping every bit in $M$. Note that if $|\supp(v)| = 1$, then $M \cdot v$ is guaranteed to return the unique coordinate where $v$ has a one, as in the proof of \Cref{lem:1-bit}. Thus, it suffices to show that we can use these queries to determine whether $|\supp(v)|$ is $0$, $1$, or strictly greater than $1$. 

First, $|\supp(v)| = 0$ iff $(M \cdot v)_1 = 0$ and $((\mathbf{1} - M) \cdot v)_1 = 0$ since the sets of $1$-coordinates in the first row of $M$ and $\mathbf{1}-M$ partition $[n]$.

Next, we claim that $|\supp(v)| > 1$ iff there exists some $i \in [\lceil \log n \rceil]$ such that $(M \cdot v)_i = 1$ and $((\mathbf{1} - M) \cdot v)_i = 1$. Note that for every row $i$, the $1$-coordinates in the $i$'th row of $M$ and $\mathbf{1}-M$ partition $[n]$. Thus, clearly if $(M \cdot v)_i = 1$ and $((\mathbf{1} - M) \cdot v)_i = 1$, then there are at least $2$ coordinates where $v$ has a one. Now we prove the converse. Suppose there exists $i \neq j \in [n]$ where $v_i = v_j = 1$. Let $b^i,b^j \in \{0,1\}^{\lceil \log n \rceil}$ denote the binary representations of $i,j$ respectively. Since $i \neq j$, there exists some bit $k$ where $b^i_k \neq b^j_k$. Without loss of generality let $b^i_k = 1$ and $b^j_k = 0$. Then,
\[
(M \cdot v)_k = \mathbf{1}\left( \left(\sum_{\ell = 1}^n v_\ell b^{\ell}\right)_k > 0\right) = \mathbf{1}\left( \sum_{\ell \colon v_{\ell} = 1}^n b^{\ell}_k > 0\right) = 1 \text{,}
\]
\[
((\mathbf{1}-M) \cdot v)_k = \mathbf{1}\left( \left(\sum_{\ell = 1}^n v_\ell (\vec{1}-b^{\ell})\right)_k > 0\right) = \mathbf{1}\left( \sum_{\ell \colon v_{\ell} = 1}^n (1-b^{\ell}_k) > 0\right) = 1
\]
and this completes the proof. \end{proof}

Next, we describe a \emph{randomized} non-adaptive algorithm for recovering the entire support of $x$.

\tER*

\begin{proof} For brevity, we assume that $t$ divides $n$. Let $m = e \cdot t \ln \frac{n}{\alpha}$. We make $\mathsf{OR}$ queries on sets $S_1, \dots, S_m$, each formed by taking $n/t$ i.i.d. uniform samples from $[n]$ and define
    \begin{align} \label{eq:X}
        X = [n]\setminus \bigcup_{\ell \in [m] \colon \mathsf{OR}_{S_{\ell}}(v)=0} S_{\ell} \text{.}
    \end{align}
If $|X| > t$, we certify $|\supp(v)| > t$ and if $|X| \leq t$, then we output $X$.



Suppose that $|\supp(v)| > t$. Observe that $\supp(v) \subseteq X$ and so $|X| > t$ with probability $1$. Thus, the algorithm is always correct in this case.

Now suppose $|\supp(v)| \leq t$. We argue that $X = \supp(v)$ with probability at least $1-\alpha$. Consider some $i \notin \supp(v)$. Note that $i \notin X$ iff there is some query $S_{\ell} \ni i$ for which $S_{\ell} \cap \supp(v) = \emptyset$. Let $\cE_{i,\ell}$ denote the event that $i \in S_{\ell}$ and $S_{\ell} \cap \supp(v) = \emptyset$. Then, since $|\supp(v)| \leq t$, we have
\begin{align}
    \Pr[\cE_{i,\ell}] = \frac{n}{t} \cdot \frac{1}{n} \cdot \left(1-\frac{|\supp(v)|}{n}\right)^{\frac{n}{t} - 1} \geq \frac{1}{t} \left(1-\frac{t}{n}\right)^{\frac{n}{t}} \geq \frac{1}{et}\nonumber
\end{align}
and so
\[ 
    \Pr[i \in X] = \Pr\big[\neg \cE_{i,\ell} \text{ for all } \ell \in [m]\big] \leq \Big(1-\frac{1}{e t}\Big)^m \leq \frac{\alpha}{N}
\]
since $m = e \cdot t \ln \frac{N}{\alpha}$. Thus, by a union bound, we have $\Pr[X \neq \supp(v)] \leq \alpha$. \end{proof}

\subsection{Connectivity Threshold of Erd\H{o}s-R\'{e}nyi Random Graphs} \label{sec:random-graph}

Our proofs in \Cref{sec:unbounded,sec:bounded-2} make use of the following bound on the probability of a random graph being connected. For intuition, note that for sufficiently large $N$, 
\[
1-(\alpha/3N)^{2/N} \approx 1 - \exp\Big(-\frac{2\ln(3N/\alpha)}{N}\Big) \approx \frac{\ln(3N/\alpha)}{N} \text{.}
\]
Thus, \Cref{fact:random-connected} asserts that for sufficiently large $N$ a random graph with $N$ vertices and $\gg N \ln (N/\alpha)$ edges is connected with probability at least $1-\alpha$, which may be a more familiar statement. However, we need such a bound to be true even for very small $N$ and so we give the following more broadly applicable version.

\randomconnected*

\begin{proof} A graph $G = (V,E)$ is connected if and only if for every cut $\emptyset \neq S \subset V$, there exists an edge $(u,v) \in E \cap (S \times \overline{S})$. When $G$ is drawn from $G(N,p)$, this does not occur for a cut $S$ of size $|S| = t$ with probability exactly $(1-p)^{t(N-t)}$. There are exactly ${N \choose t}$ such cuts. Thus, taking a union bound over all cuts and using our lower bound on $p$, we have
\begin{align}
    \Pr_{G \sim G(N,p)}[G \text{ not connected}] &\leq \sum_{t = 1}^{N-1} {N \choose t} \left(\frac{\alpha}{3N}\right)^{\frac{2}{N} \cdot t(N-t)} \nonumber \\
    &\leq 2\sum_{t = 1}^{\floor{N/2}} {N \choose t} \left(\frac{\alpha}{3N}\right)^{\frac{2}{N} \cdot t(N-t)} \nonumber \\
    &\leq 2\sum_{t = 1}^{\floor{N/2}} {N \choose t} \left(\frac{\alpha}{3N}\right)^{\frac{2}{N} \cdot \frac{tN}{2}} \leq 2\sum_{t = 1}^{\floor{N/2}} N^t \left(\frac{\alpha}{3N}\right)^{t} = 2 \sum_{t = 1}^{\floor{N/2}} \left(\alpha/3\right)^t \leq \alpha \nonumber
\end{align}   
and this completes the proof. \end{proof}

\section{A Simple $O(n\log{k})$ Adaptive Algorithm}
\label{sec:simple-adaptive}
Here we sketch a simple adaptive algorithm using $O(n\log{k})$ queries. Suppose, we have identified one element from $i$ clusters (initially $i=0$, and we have $i \leq k$ always). Suppose they are $X=\{x_1, x_2,...,x_i\}$. We now want to find the cluster to which a new point $y$ belongs to. We first query $X \cup \{y\}$. If the answer is $i+1$, then $y$ is part of a new cluster and $i$ grows to $i+1$. Otherwise, $y$ is part of the $i$ clusters, and we detect the cluster to which $y$ belongs to using a binary search. We consider the two sets $X_1=\{x_1, x_2,..,x_{\ceil{i/2}}\}$, and $X_2=\{x_{\ceil{i/2}+1},..,x_i\}$. We then query $X_1 \cup \{y\}$. If the answer is $\ceil{{i/2}}+1$, then we search recursively in $X_2$, else if the query answer is $\ceil{{i/2}}$, then we search recursively in $X_1$. Clearly, the query complexity is $O(\log{k})$ per item, and it requires $O(\log{k})$ rounds of adaptivity even to place one element.

\end{document}